\documentclass[11pt,twoside]{article}
\usepackage{amsmath, amssymb, amsthm, amsfonts,cite,alltt,clrscode}
\usepackage[english]{babel}
\usepackage{graphicx}
\usepackage{caption}
\usepackage{url}
\usepackage{float}
\usepackage{wrapfig}
\usepackage{array}
\usepackage{txfonts}
\usepackage{hyperref}
\usepackage[list=true,listformat=simple]{subcaption}

\usepackage{cmap}
\usepackage[T1]{fontenc}

\newtheorem{theorem}{Theorem}[section]

\graphicspath{{./figuresPushPull/}}

{\makeatletter
 \gdef\xxxmark{%
   \expandafter\ifx\csname @mpargs\endcsname\relax 
     \expandafter\ifx\csname @captype\endcsname\relax 
       \marginpar{xxx}
     \else
       xxx 
     \fi
   \else
     xxx 
   \fi}
 \gdef\xxx{\@ifnextchar[\xxx@lab\xxx@nolab}
 \long\gdef\xxx@lab[#1]#2{\textbf{[\xxxmark #2 ---{\sc #1}]}}
 \long\gdef\xxx@nolab#1{\textbf{[\xxxmark #1]}}
}

\let\realbibitem=\bibitem
\def\bibitem{\par \vspace{-1.2ex}\realbibitem}

\usepackage{times}

\newlength\aboveparagraphskip
\newlength\belowparagraphskip
\makeatletter
\def\paragraph{\@startsection{paragraph}{4}{\z@}{-\aboveparagraphskip}%
                 {\belowparagraphskip}{\normalfont\normalsize\bfseries}}
\makeatother
\def\compactify{\itemsep=0pt \topsep=0pt \partopsep=0pt \parsep=0pt}
\let\latexusecounter=\usecounter
\newenvironment{itemize*}
  {\begin{itemize}\compactify}
  {\end{itemize}}
\newenvironment{enumerate*}
  {\def\usecounter{\compactify\latexusecounter}
   \begin{enumerate}}
  {\end{enumerate}\let\usecounter=\latexusecounter}
\newenvironment{description*}
  {\begin{description}\compactify}
  {\end{description}}

\let\epsilon=\varepsilon

\usepackage[margin=1in]{geometry}

\begin{document}

\title{Push-Pull Block Puzzles are Hard}
\author{Erik D. Demaine\thanks{MIT Computer Science and Artificial Intelligence Laboratory, 32 Vassar Street, Cambridge, MA 02139, USA, \url{{edemaine,isaacg,jaysonl}@mit.edu}} \and Isaac Grosof\footnotemark[1]  \and Jayson Lynch\footnotemark[1]}
\date{}
\maketitle

\begin{abstract}
This paper proves that push-pull block puzzles in 3D are PSPACE-complete to solve, and push-pull block puzzles in 2D with thin walls are NP-hard to solve, settling an open question \cite{zubaranagent}. Push-pull block puzzles are a type of recreational motion planning problem, similar to Sokoban, that involve moving a `robot' on a square grid with $1 \times 1$ obstacles. The obstacles cannot be traversed by the robot, but some can be pushed and pulled by the robot into adjacent squares. Thin walls prevent movement between two adjacent squares. This work follows in a long line of algorithms and complexity work on similar problems \cite{PushPull91,Push100,Push*00,PushPushk04,non-crossing01,DO92,Push2F02,Sokoban98,DZ96,Pull10}. The 2D push-pull block puzzle shows up in the video games \emph{Pukoban} as well as \emph{The Legend of Zelda: A Link to the Past}, giving another proof of hardness for the latter \cite{NintendoFun2014}. This variant of block-pushing puzzles is of particular interest because of its connections to reversibility, since any action (e.g., push or pull) can be inverted by another valid action (e.g., pull or push).
\end{abstract}

%
%
%

%
%
\section{Introduction} 
Block-pushing puzzles are a common puzzle type with one of the best known example being \emph{Sokoban}. Puzzles with the ability to push and pull blocks have found their way into several popular video games including \emph{The Legend of Zelda} series, \emph{Starfox Adventures}, \emph{Half-Life} and \emph{Tomb Raider}. Block-pushing puzzles are also an abstraction of motion planning problems with movable obstacles. In addition to these games, one could imagine real-world scenarios, like that of a forklift in a warehouse, bearing similarity. Since motion planning is such an important and computationally difficult problem, it can be useful to look at simplified models to try to get a better understanding of the larger problem.

A significant amount of research has gone into characterizing the complexity of block sliding puzzles. This includes PSPACE-completeness for well-known puzzles like sliding-block puzzles\cite{hearn2005pspace}, Sokoban \cite{Sokoban98, DZ96}, the 15-puzzle \cite{15Puzzle}, 2048\cite{abdelkader2048}, Candy Crush\cite{guala2014bejeweled} and Rush Hour \cite{RushHour02}. Block pushing puzzles are a type of block sliding puzzle in which the blocks are moved by a small robot within the puzzles. This type of block sliding puzzle has gathered a significant amount of study. Table~\ref{BlocksTable} gives a summary of results on block pushing puzzles. Variations include Sokoban\cite{Sokoban98, DZ96}, where blocks must reach specific targets (the \emph{Path?} column), versions where multiple blocks can be pushed\cite{Push100, Push*00, Push2F02, DZ96, Sokoban98, Pull10} (the \emph{Push} column), versions where blocks continue to slide after being pushed\cite{PushPushk04, Push*00} (the \emph{Sliding} column), versions where fixed blocks are allowed\cite{DO92, Push2F02} (the \emph{Fixed?} column), and versions where the robot can pull blocks\cite{Pull10} (the \emph{Pull} column). 

We are particularly interested in the push-pull block model because any sequence of moves in the puzzle can be undone.  Having an undirected state-space graph seems like an interesting property both mathematically and from a puzzle stand-point. This sort of player move reversibility lead to some of our gadgets being logically reversible, a notion that is fundamentally linked to quantum computation and the thermodynamics of computation. 

\begin{table}
\centering
\centerline{
\begin{tabular}{|l|l|l|l|l|l|l|l|l|}
\hline
\emph{Name} & \emph{Push} & \emph{Pull} & \emph{Fixed?} & \emph{Path?} & \emph{Sliding} & \emph{Complexity} \\ \hline
\hline
Push-$k$ & $k$ & 0 & No & Path & Min & NP-hard\cite{Push100} \\ \hline
Push-$*$ & $*$ & 0 & No & Path  & Min & NP-hard\cite{Push*00} \\ \hline
PushPush-$k$ & $k$ & 0 & No & Path  & Max & PSPACE-c.\cite{PushPushk04} \\ \hline
PushPush-$*$ & $*$ & 0 & No & Path  & Max & NP-hard\cite{Push*00} \\ \hline
Push-$1$F & $1$ & $0$ & Yes & Path  & Min &  NP-hard \cite{DO92} \\ \hline
Push-$k$F & $k\geq 2$ & $0$ & Yes & Path  & Min & PSPACE-c.\cite{Push2F02} \\ \hline
Push-$*$F & $*$ & $0$ & Yes & Path  & Min & PSPACE-c.\cite{Push2F02} \\ \hline
Sokoban & $1$ & $0$ & Yes & Storage  & Min & PSPACE-c.\cite{Sokoban98} \\ \hline
Sokoban$(k,1)$ & $k\geq 5$ & $1$ & Yes & Storage  & Min & NP-hard\cite{DZ96} \\ \hline
Pull-$1$ & $0$ & $1$ & No & Storage  & Min & NP-hard\cite{Pull10} \\ \hline
Pull-$k$F & $0$ & $k$& Yes & Storage  & Min &  NP-hard\cite{Pull10} \\ \hline
PullPull-$k$F & $0$ & $k$ & Yes & Storage  & Max  & NP-hard\cite{Pull10} \\ \hline
\textbf{Push-$k$ Pull-$l$W} & $k$ & $l$ & Wall & Path  & Min & \textbf{NP-hard}\ (\S  \ref{2DNPhard}) \\ \hline
\textbf{3D Push-$k$ Pull-$l$F} & $k$ & $l$ & Yes & Path & Min &  \textbf{NP-hard}\ (\S  \ref{3DNPhard}) \\ \hline
\textbf{3D Push-$1$ Pull-$1$W} & $1$ & $1$ & Wall & Path & Min &  \textbf{PSPACE-c.}\ (\S  \ref{3DPSPACE}) \\ \hline
\textbf{3D Push-$k$ Pull-$k$F} & $k > 1$ & $k >1$ & Yes & Path & Min &  \textbf{PSPACE-c.}\ (\S  \ref{3DPSPACE}) \\ \hline
\end{tabular}
}
\caption{Summary of past and new results on block pushing and/or pulling. The \emph{Push} and \emph{Pull} columns describe how many blocks in a row can be moved by the robot. Here $k$ and $l$ are positive integers; $*$ refers to an unlimited number of blocks. The \emph{Fixed} column notes whether fixed blocks (Yes) or thin walls (Wall) are allowed. In the problem title, F means fixed blocks are included; W means thin walls are included. The \emph{Path} column describes whether the objective is to have the robot find a path to a target location, or to store the blocks in a specific configuration. The \emph{Sliding} column notes whether blocks move one square or as many squares as possible before stopping.}
\label{BlocksTable}
\end{table}

We add several new results showing that certain block pushing puzzles, which include the ability to push and pull blocks, are NP-hard or PSPACE-complete. The push-pull block puzzle is instantiated in the game Pukoban and heuristics for solving it have been studied \cite{zubaranagent}, but its computational complexity was left as an open question.

We introduce \emph{thin walls}, which prevent motion between two adjacent empty squares. We prove that all path planning problems in 2D with thin walls or in 3D, in which the robot can push $k$ blocks and pull $l$ blocks for all $k,l \in \mathbb{Z}^+$ are NP-hard. We also show that path planning problems where the robot can push and pull $k$ blocks are PSPACE-complete, with thin walls needed only for $k=1$. Our results are shown in the last four lines of Table~\ref{BlocksTable}. To prove these results, we introduce two new abstract gadgets, the set-verify and the 4-toggle, and prove hardness results for questions about their the legal state transitions. 

2D Push-$k$ Pull-$j$ is defined as follows: There is a square lattice of cells. Each cell is connected to its orthogonal neighbors. Cells may either be empty, hold a movable block, or hold a fixed block. Additionally, in settings that allow thin walls, edges between cells may be omitted. There is also a robot on a cell. The robot may move from its current cell to an unoccupied adjacent cell. The robot may also \emph{push} up to $k$ movable blocks arranged in a straight line one cell forward, as long as there is an open cell with no wall in that direction. Here the robot moves into the cell occupied by the adjacent block and each subsequent block moves into the adjacent cell in the same direction. Likewise, the robot may \emph{pull} up to $j$ movable blocks in a straight line as long as there are no walls in the way and there is an open cell behind the robot. The robot moves into that cell, the block opposite that cell moves into the one the robot originally occupied, and subsequent blocks also move once cell toward the robot. The goal of the puzzle is for the robot to reach a specified goal cell. Given such a description, is there a legal path for the robot from its starting cell to the goal cell? The 3D problem is defined analogously on a cubic lattice.

%
\section{Push-Pull Block Puzzles are NP-hard}
\label{2DNPhard}
We show NP-hardness for Push-$k$ Pull-$l$ in 2D with thin walls for all positive integers $k$, $l$ in Section~\ref{2DNPhard} and Push-$q$ Pull-$r$ in 3D for all positive integers $q, r$ in Section~\ref{3DNPhard}.

\subsection{2D Push-Pull with Thin Walls}
\label{2DNPhard}
In this section we prove that Push-$k$ Pull-$l$ in 2D with fixed blocks is NP-hard, for all positive $k$ and $l$, if we include \emph{thin walls}. 
Thin walls are a new, but natural, notion for block pushing puzzles. They prevent blocks or the robot from passing between two adjacent, empty squares, as though there were a thin wall blocking the path. We will prove hardness by a reduction from 3SAT. The 3SAT problem asks whether, given a set of variables $\{x_1, x_2, \ldots x_n\}$ and a boolean formula in conjunctive normal form with exactly three variables per clause, there exists an assignment of values to those variables that satisfies the formula\cite{NPBook}. To do so we will introduce an abstract gadget called the Set-Verify gadget. This gadget will then be used to construct crossover gadgets (in Appendix~\ref{sec:NPCrossover}), and variable and clause gadgets.

\subsubsection{Set-Verify Gadgets}
\label{sec:SetVerifyGadgets}
The Set-Verify gadget is an abstract gadget for motion planning problems. The gadget has four entrances/exits which have different allowable paths between them depending on the state of the gadget. There are four possible states of the Set-Verify gadget: Broken, Unset, Set, and Verified. The three relevant states are depicted in Figures~\ref{setVerifyDiagrams} and \ref{SetVerifyStateTransition}. Entrances to the gadget are labeled $S_i, S_o, V_i, V_o$ and the directed arrows show the allowed passages in the shown state. Further details are given in Appendix~\ref{sec:NPCrossover}.
In the Unset state, the $S_i \rightarrow S_o$ transition is the only possibility, changing the state to Set. In the Set state, the $S_o \rightarrow S_i$ transition is possible, changing the state back to Unset, as well as the $V_i \rightarrow V_o$ transition, which changes the state to Verified. Finally, from the Verified state, the only transitions possible are $V_o \rightarrow V_i$, changing the state back to Set, and $V_i \rightarrow V_o$, leaving the state as Verified. In the Broken state, the only possible transition is $S_o \rightarrow S_i$, changing the state to Unset. Any time we would enter the Broken state, we could instead enter the Set state, which allows strictly more transitions, and therefore will be strictly more helpful in reaching the goal. The Broken state is not helpful towards reaching the goal, so we will disregard its existence.

For the Set-Verify gadget in the Unset state, the $S_i$ entrance is the only one which allows the robot to move any blocks. From the $S_i$ entrance it can traverse to $S_o$, and it can also pull block $2$ down behind them. Doing so will allow a traversal from $V_i$ to $V_o$. To traverse back from $S_o$ to $S_i$, the robot must first traverse back from $V_o$ to $V_i$. Then, when the robot travels back from $S_o$ to $S_i$, it must push block $2$ back, ensuring the $V_i$ to $V_o$ traversal is impossible. Further, access to any sequence of entrances will not allow the robot to alter the system to allow traversals between the $V_i$ and $S_i$ entrances. 

Since the Set-Verify gadget has no hallways with length greater than $3$, any capabilities the robot may have of pushing or pulling more than one block at a time are irrelevant. Thus, the following proof will apply for all positive values of $j$ and $k$ in Push-$j$ Pull-$k$.

\begin{figure}[!ht]
  \centering
    \begin{subfigure}[b]{0.3\textwidth}
    \includegraphics[width=\textwidth]{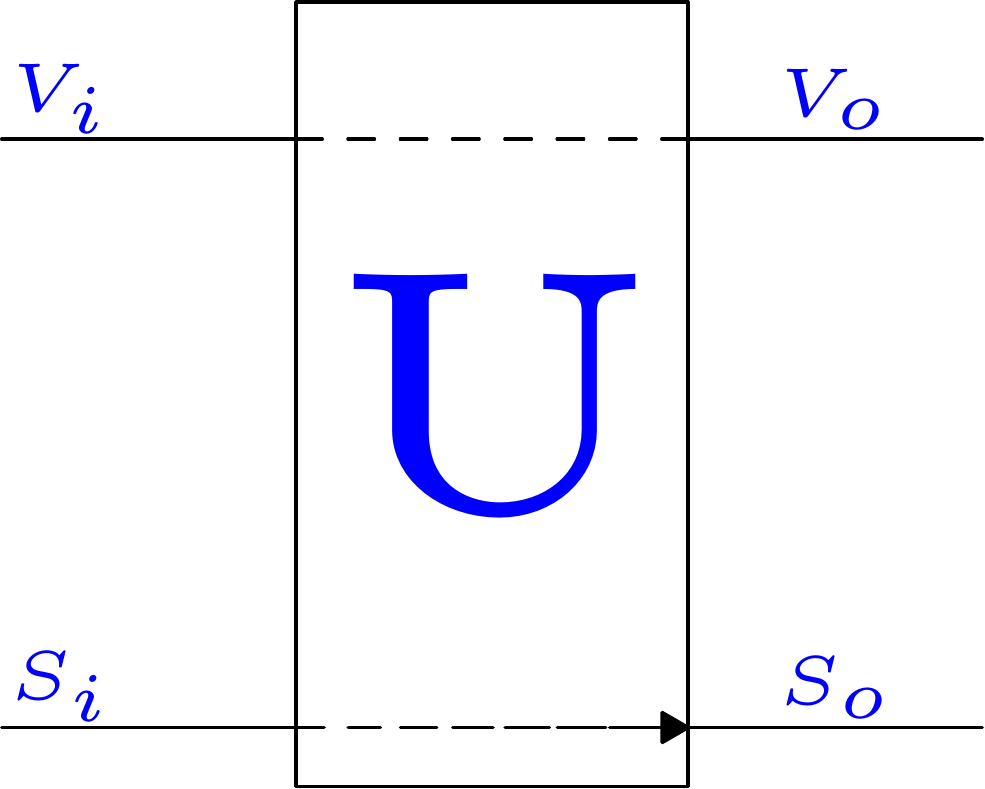}
    \caption{Abstract Unset Set-Verify}
    \vspace{15pt}
    \end{subfigure}
    \hfill
    \begin{subfigure}[b]{0.3\textwidth}
    \includegraphics[width=\textwidth]{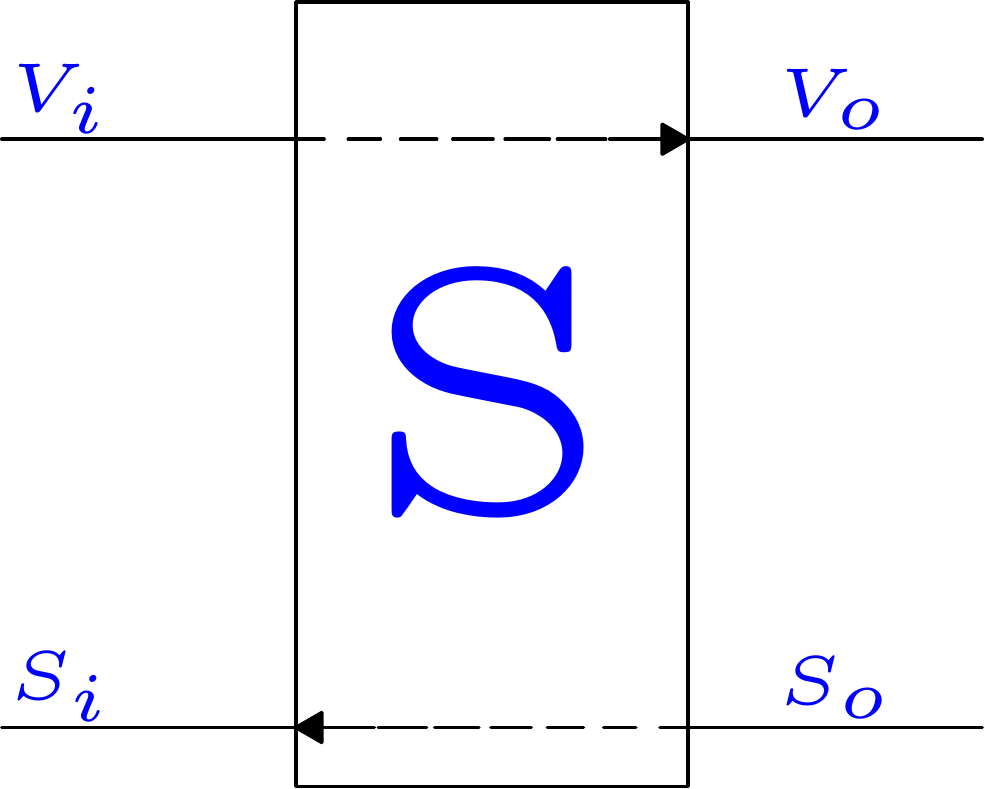}
    \caption{Abstract Set Set-Verify}
    \vspace{15pt}
    \end{subfigure}
    \hfill
    \begin{subfigure}[b]{0.33\textwidth}
    \includegraphics[width=.9\textwidth]{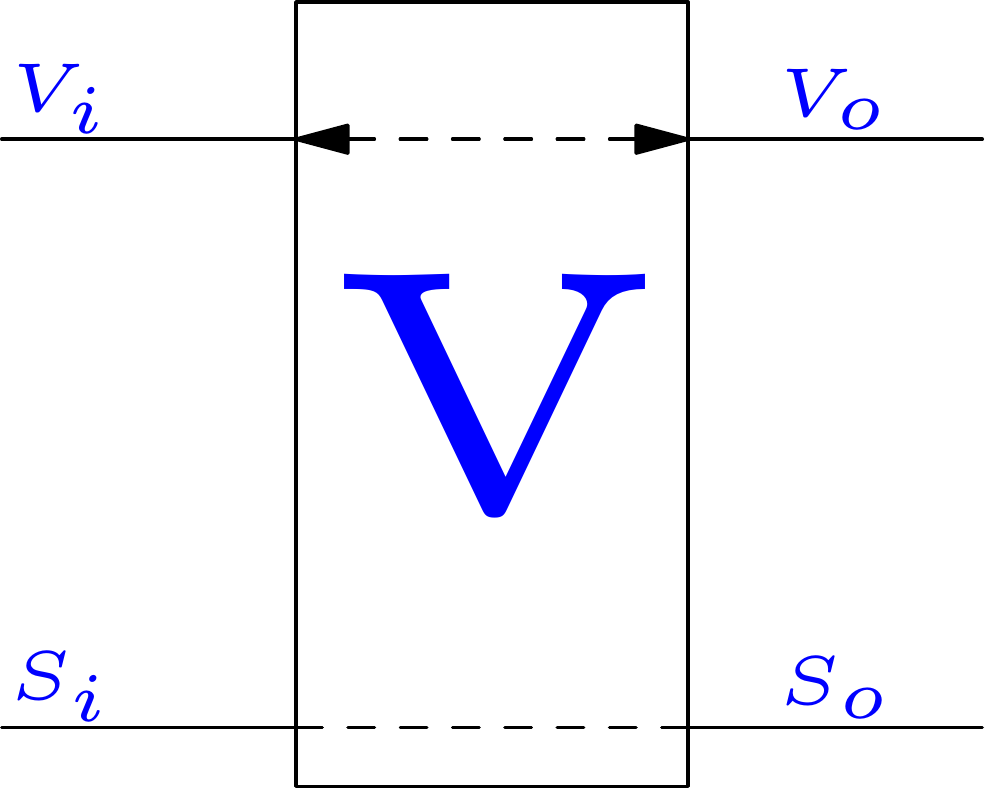}
    \caption{Abstract Verified Set-Verify}
    \vspace{15pt}
  \end{subfigure}  
   \vspace{10pt}
  \begin{subfigure}[b]{0.3\textwidth}
    \includegraphics[width=\textwidth]{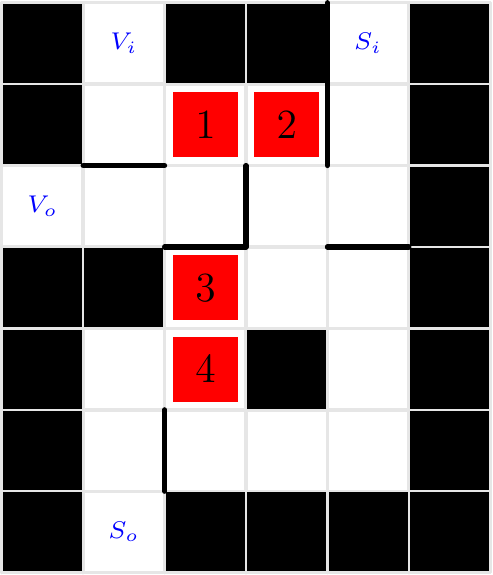}
    \caption{Set-Verify, unset state}
    \label{SetVerifyUnset}
  \end{subfigure}
  \hfill
  \begin{subfigure}[b]{0.3\textwidth}
    \includegraphics[width=\textwidth]{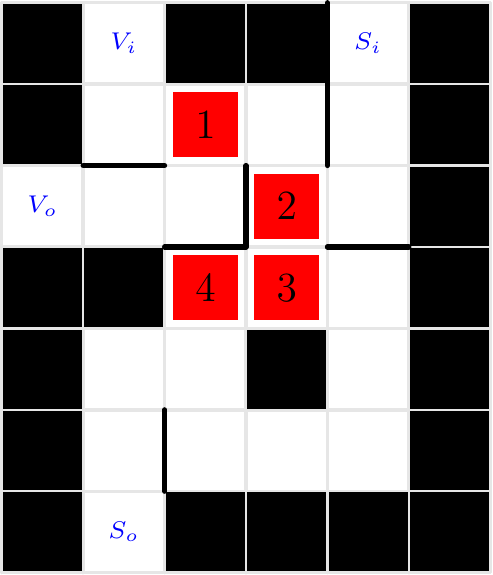}
    \caption{Set-Verify, set state}
    \label{SetVerifySet}
  \end{subfigure}
  \hfill
  \begin{subfigure}[b]{0.3\textwidth}
    \includegraphics[width=\textwidth]{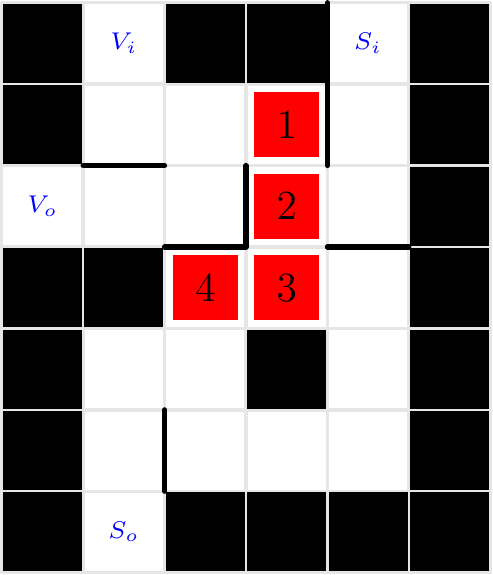}
    \caption{Set-Verify, verified state}
    \label{SetVerifyVerified}
  \end{subfigure}
  \caption{Diagrams of three of the states of Set-Verify gadgets along with their construction in a push-pull block puzzle. Red blocks are moveable, black blocks are fixed, thick black lines are thin walls.}
  \label{setVerifyDiagrams}
\end{figure}

\begin{table}
\begin{minipage}{.3\textwidth}
{\setlength\tabcolsep{4pt}
\begin{tabular}{>{$} l <{$} >{$} l <{$} >{$} l <{$} >{$} l <{$}}
   U: & & & \\
   &(U, s_i)& \rightarrow& (S, S_o) \\
\end{tabular}}
\end{minipage}
\begin{minipage}{.3\textwidth}
{\setlength\tabcolsep{4pt}
\begin{tabular}{>{$} l <{$} >{$} l <{$} >{$} l <{$} >{$} l <{$}}
  S: & & & \\
   &(S, s_o)& \rightarrow& (U, S_i) \\
   &(S, v_i)& \rightarrow& (V, v_o) \\
\end{tabular}}
\end{minipage}  
\begin{minipage}{.3\textwidth}
{\setlength\tabcolsep{4pt}
\begin{tabular}{>{$} l <{$} >{$} l <{$} >{$} l <{$} >{$} l <{$}}
   V: & & & \\
   &(V, v_i)& \rightarrow& (V, v_o) \\ 
   &(V, v_o)& \rightarrow& (V, v_i) \\ 
   &(V, v_o)& \rightarrow& (S, v_i) \\ 
\end{tabular}}
\end{minipage}
\caption{State transitions of a Set-Verify gadget as seen in Figure~\ref{setVerifyDiagrams}}
\label{SetVerifyStateTransition}
\end{table}

\subsubsection{Variable and Clause Gadgets}
\label{sec:2DPushPull3SAT}

We will be making use of the Set-Verify gadget to produce the literals in our 3SAT formula. One significant difficulty with this model is the complete reversibility of all actions. Thus we need to take care to ensure that going backward at any point does not allow the robot to cheat in solving our 3SAT instance. The directional properties of the Set-Verify allow us to create sections where we know if the robot exits, it must have either reset everything to the initial configuration or have correctly proceeded through that gadget.

Our literals will be represented by Set-Verify gadgets. They are considered true when the $V_i$ to $V_o$ traversal is possible, and false otherwise. Thus we can set literals to true by allowing the robot to run through the $S_i$ to $S_o$ passage of the gadget. This allows a simple clause gadget, shown in Figure~\ref{fig:NPClauseGadget}, consisting of splitting the path into three hallways, each with the corresponding verify side of our literal. We can then pass through if any of the literals are set to true, and cannot pass otherwise. Notice that the Unset and Set states do not have a backward transition. Thus the only way to go back through the clause is through the verified literal, after which the clause has been reset to the state it was in before the robot went through it.

The variables will be encoded by a series of passages which split to allow either the true or negated literals to be set, shown in Figure~\ref{fig:NPVariableGadget}. Once the robot has gone through at least one gadget in one hallway, there are only two possibilities remaining: either the robot can continue down the hall setting more literals to true, or the robot can go back through the gadget it has just exited, returning it to its unset state. Thus, before entering or after exiting a hallway all of the literals in that hallway will be in the same state. Additionally, unset gadgets do not allow a transition from $S_o$ to $S_i$, which means at any point while setting variables, if the robot decides to go back it can only return through a hallway which has been switched to the set state. Going back through these returns them to the unset state, putting that variable gadget back in its initial configuration before the robot interacted with it.

\begin{figure}[!ht]
\begin{minipage}{.36\textwidth}
    \includegraphics[width=\textwidth]{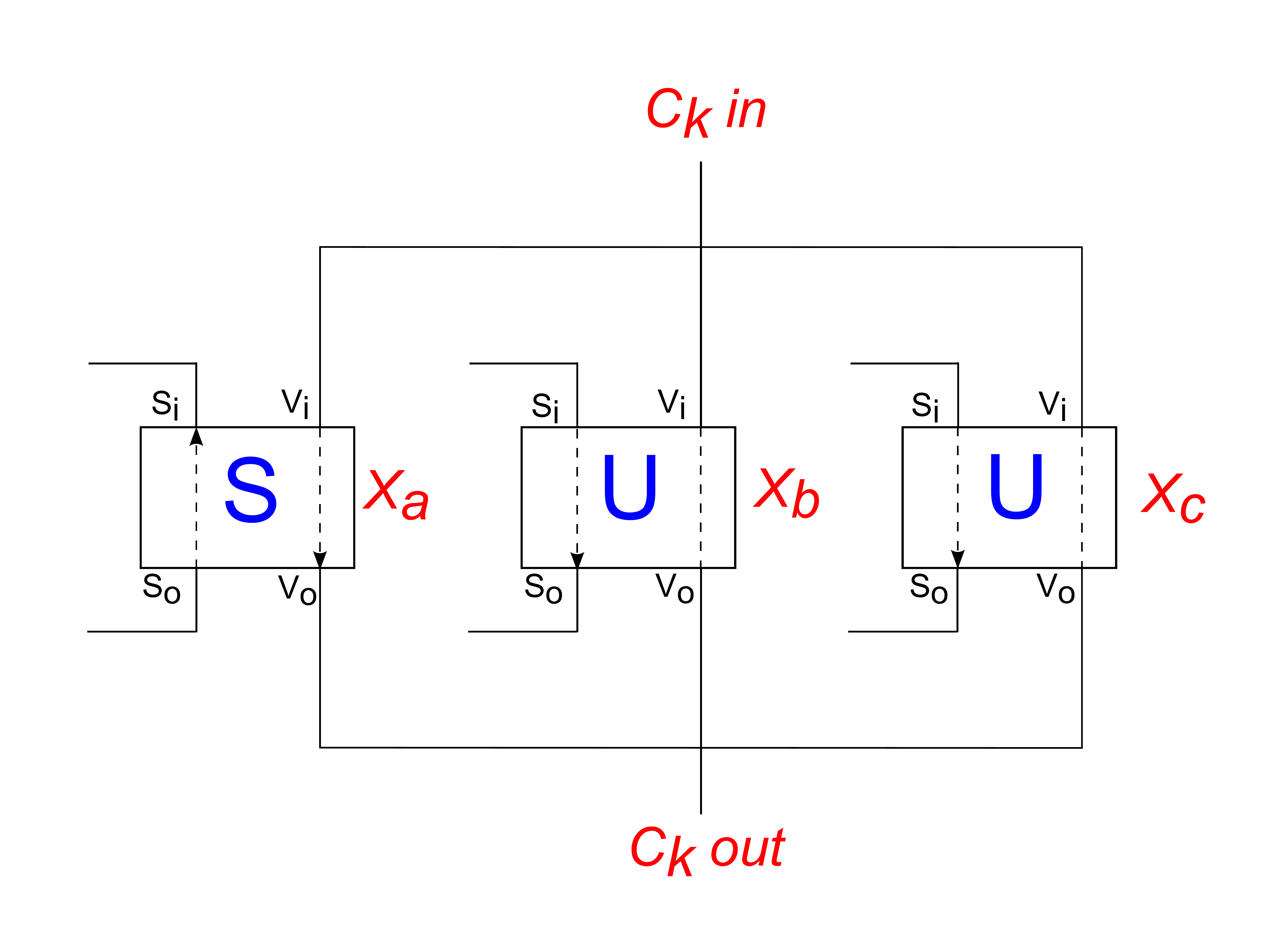}
    \caption{Clause gadget, $C_k$, with variables $x_a=1$, $x_b=0$, $x_c=0$.}
    \label{fig:NPClauseGadget}
\end{minipage}
\hspace{5mm}
\begin{minipage}{.57\textwidth}
  \centering
    \includegraphics[width=.8\textwidth]{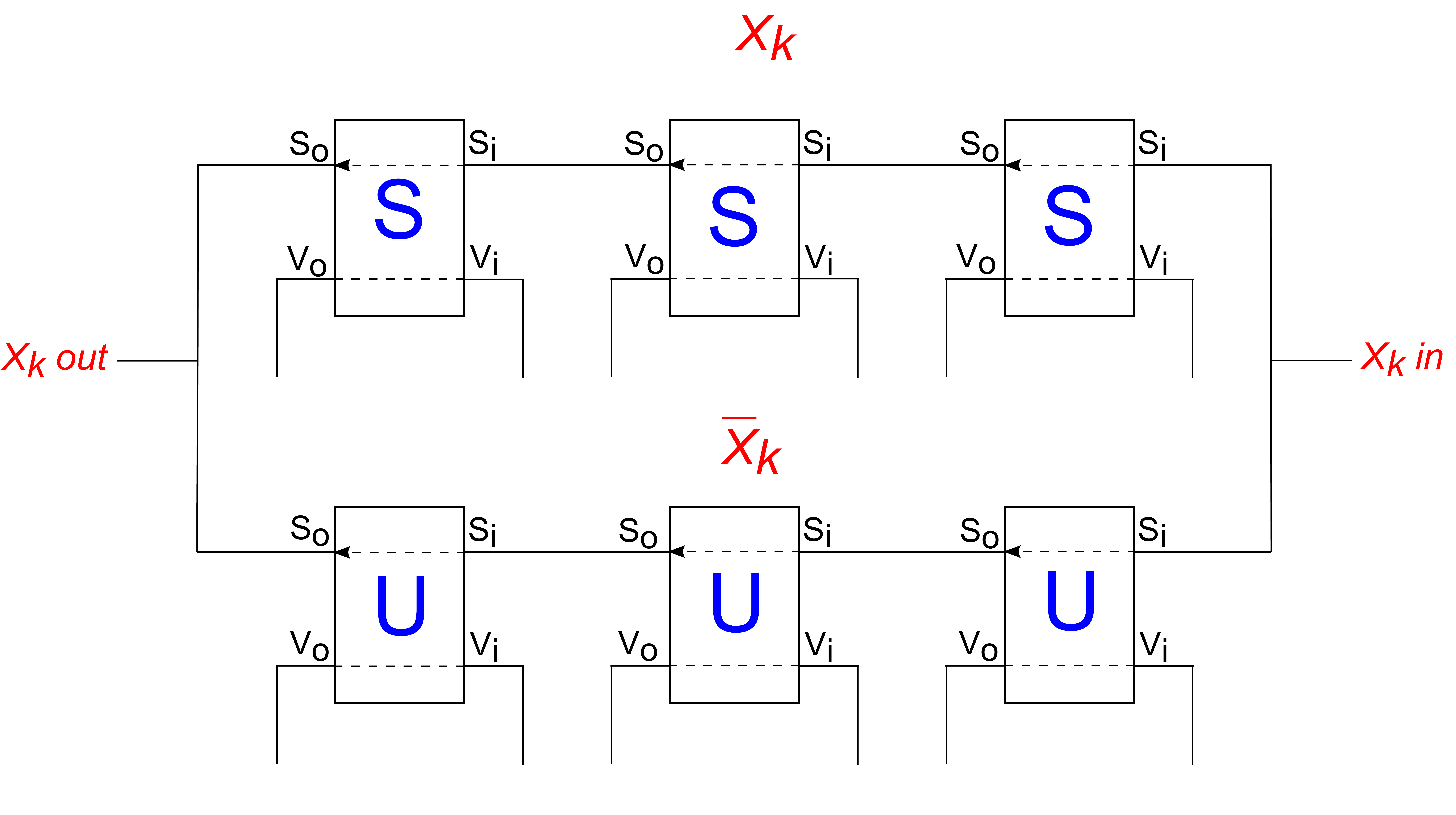}
    \caption{A variable gadget representing $X_k$ occurring in six clauses, three of those times negated. The value of the variable has been set to true.}
    \label{fig:NPVariableGadget}
\end{minipage}
\end{figure}

%
%
\subsubsection{Crossover Gadgets}
\label{sec:NPCrossover}
In this section we build up the needed two use crossover gadget from a series of weaker types of crossover gadgets. One may wonder why we need crossover gadgets when Planar 3SAT is NP-complete. This only guarantees that connecting the verticies to their clauses by edges results in a planar graph, it does not ensure that we can navigate our robot between all of these gadgets in a planar manner or that our gadgets themselves are planar. The most obvious issue can be seen in the clause gadget (Figure~\ref{fig:NPClauseGadget}) where one of the Set Verify gadgets must lie between the other two hallways, but must also be accessible by its associated variable gadget.

\paragraph{Directed Destructive Crossover} This gadget, depicted in Figure~\ref{fig:DestructiveCrossover}, allows either a traversal from $a$ to $a'$ or $b$ to $b'$. Once a traversal has occurred, that path may be traversed in reverse, but the other is impassable unless the original traversal is undone.

First, observe that transitions are initially only possible via the $a$ and $b$ entrances, since the transitions possible through a Set-Verify in the Set state can be entered through $V_i$ and $S_o$, not $S_i$. Assume without loss of generality that the gadget is entered at $a$. This changes the state of the left Set-Verify to Verified. At this point, only the right $S_o$ and left $V_o$ transitions are passable. Taking the $V_o$ transition either reverts all changes to the original state, or leaves the left crossover in the Verified state, which allows strictly less future transition than the original state. Therefore, we will disregard that option. Taking the $S_o$ transition changes the right Set-Verify to Unset, and completes the crossover. At this point, the only possible transition is to undo the transition just made, from $a'$ back to $a$, restoring the original state. The gadget could be entered via $a$, but the robot would only be able to leave via $a$, possibly changing the state to Set. Both options result in the robot exiting out its original entrance, and allow the same or less future transtions, so we may disregard those options. Thus, the only transition possibilities are as stated above.

\begin{figure}[!ht]
  \centering
  \begin{subfigure}[b]{0.47\textwidth}
    \includegraphics[width=\textwidth]{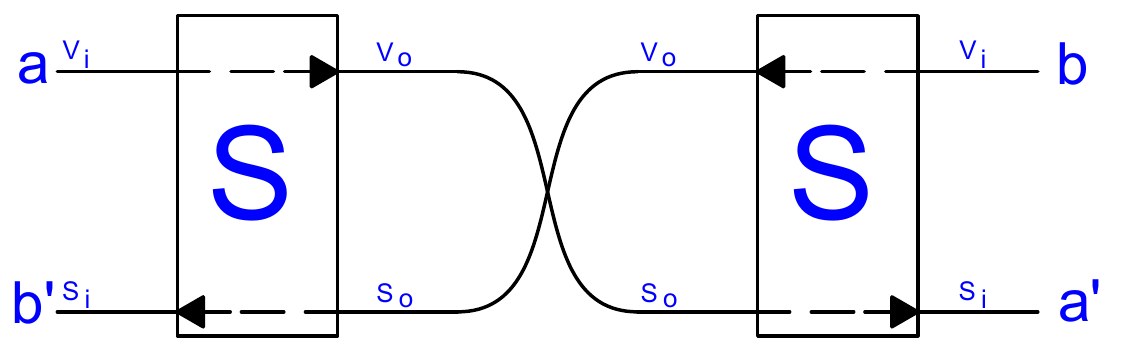}
    \caption{The directed destructive crossover constructed from two connected Set-Verify gadgets initialized in the set position.}
    \label{fig:DestructiveCrossover}
  \end{subfigure}
  \hfill
  \begin{subfigure}[b]{0.47\textwidth}
    \includegraphics[width=\textwidth]{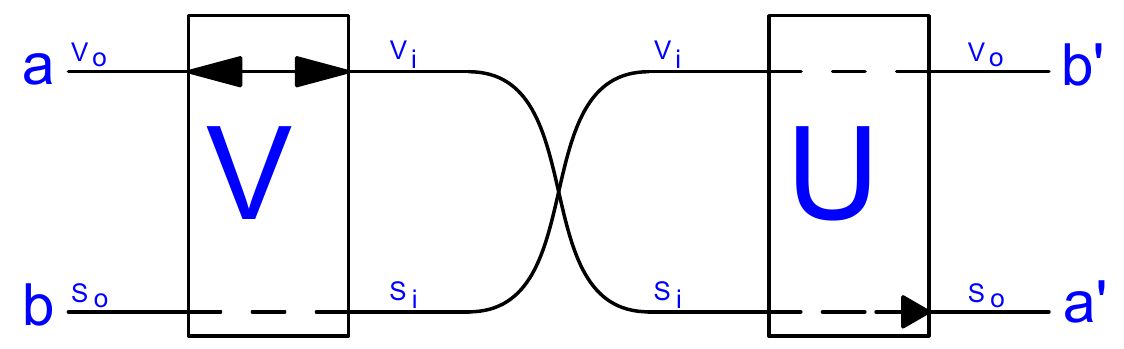}
    \caption{The in-order directed crossover constructed from two connected Set-Verify gadgets initialized in the verified and unset positions.}
    \label{fig:InOrderCrossover}
  \end{subfigure}
  \caption{Two types of crossover gadgets}
\end{figure}

\begin{figure}[!ht]
  \centering
  \begin{subfigure}[b]{0.48\textwidth}
    \includegraphics[width=\textwidth]{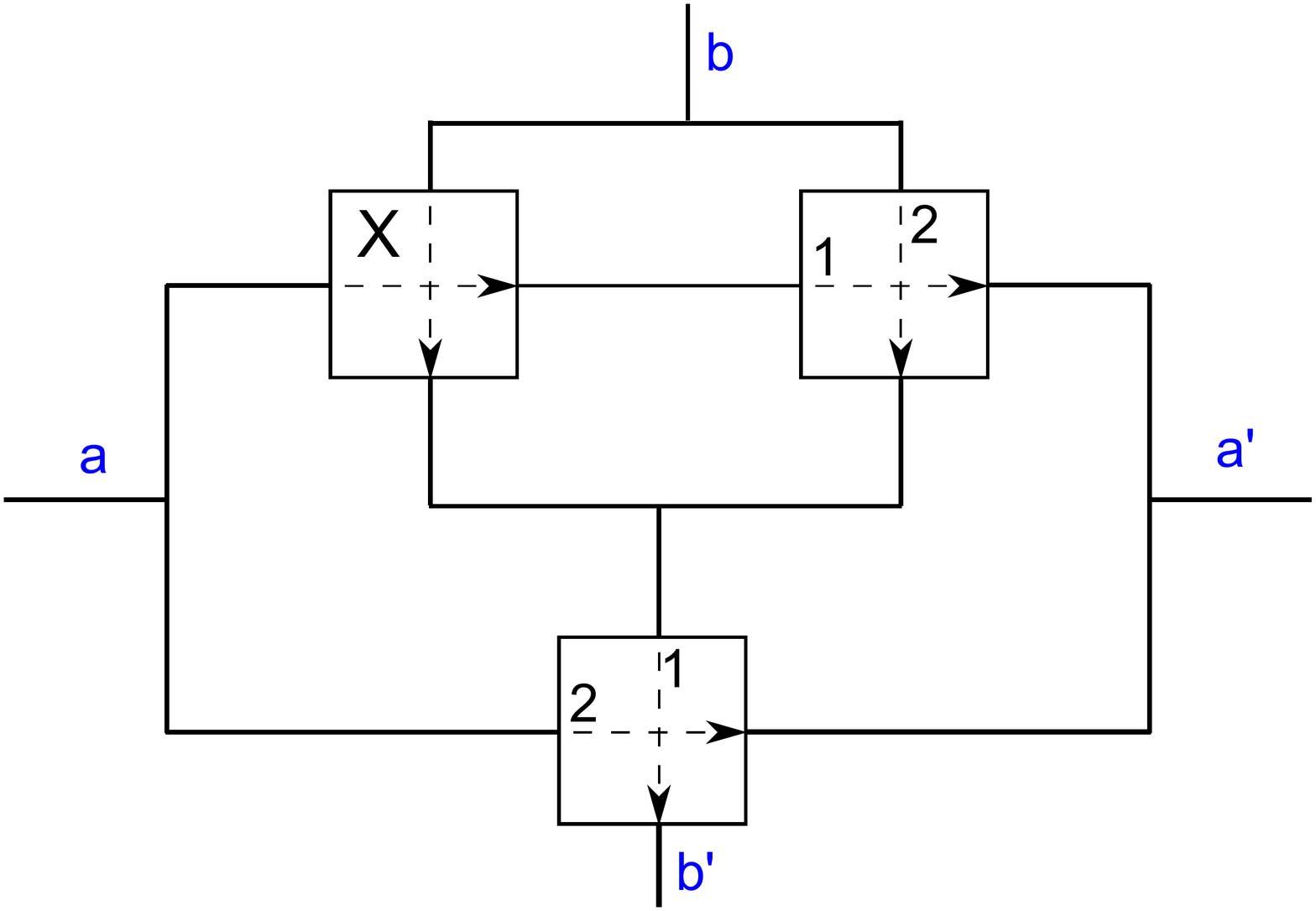}
    \caption{The one use directed crossover is constructed from a directed destructive crossover and two in-order directed crossovers.}
    \label{OneUseCrossover}
  \end{subfigure}
  \hfill
  \begin{subfigure}[b]{0.43\textwidth}
    \includegraphics[width=\textwidth]{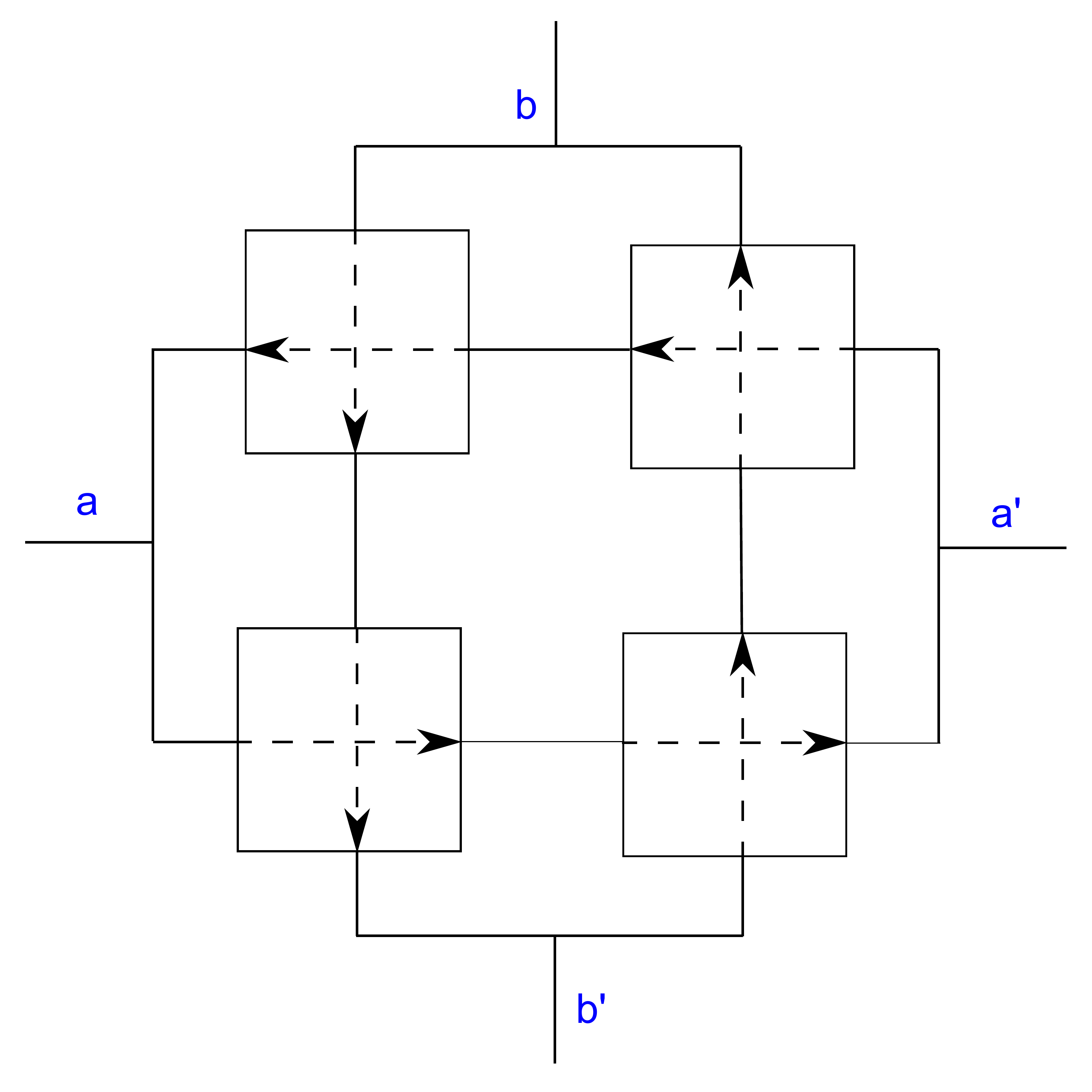}
    \caption{A full one use crossover constructed from four directed one use crossovers.}
    \label{full_one_use_crossover}
  \end{subfigure}
  \caption{Composite crossover gadgets}
\end{figure}

\paragraph{In-order Directed Crossover} This gadget, depicted in Figure~\ref{fig:InOrderCrossover} allows a traversal from $a$ to $a'$, followed by a traversal from $b$ to $b'$. These traversals may also be reversed.

Initially, no entrance is passable except for $a$, since $V_o$ is passable only in the Verified state, and $S_o$ is
passable only in the Set state. Once the left $V_o \rightarrow V_i$ transition is made, the robot has 2 options.
It can either change the left Set-Verify gadget's state to Set, or leave it as Verified. In either case, the $S_i$
entrance on that toggle is impassable, since a $S_i$ entrance may only be traversed in the Unset state. The
only transition possible on the right crossover is $S_i \rightarrow S_o$, changing the state from Unset to Set.
This completes the first crossing.

\begin{wrapfigure}{hr}{0.45\textwidth}
\vspace{-8mm}
  \centering
    \includegraphics[width=.45\textwidth]{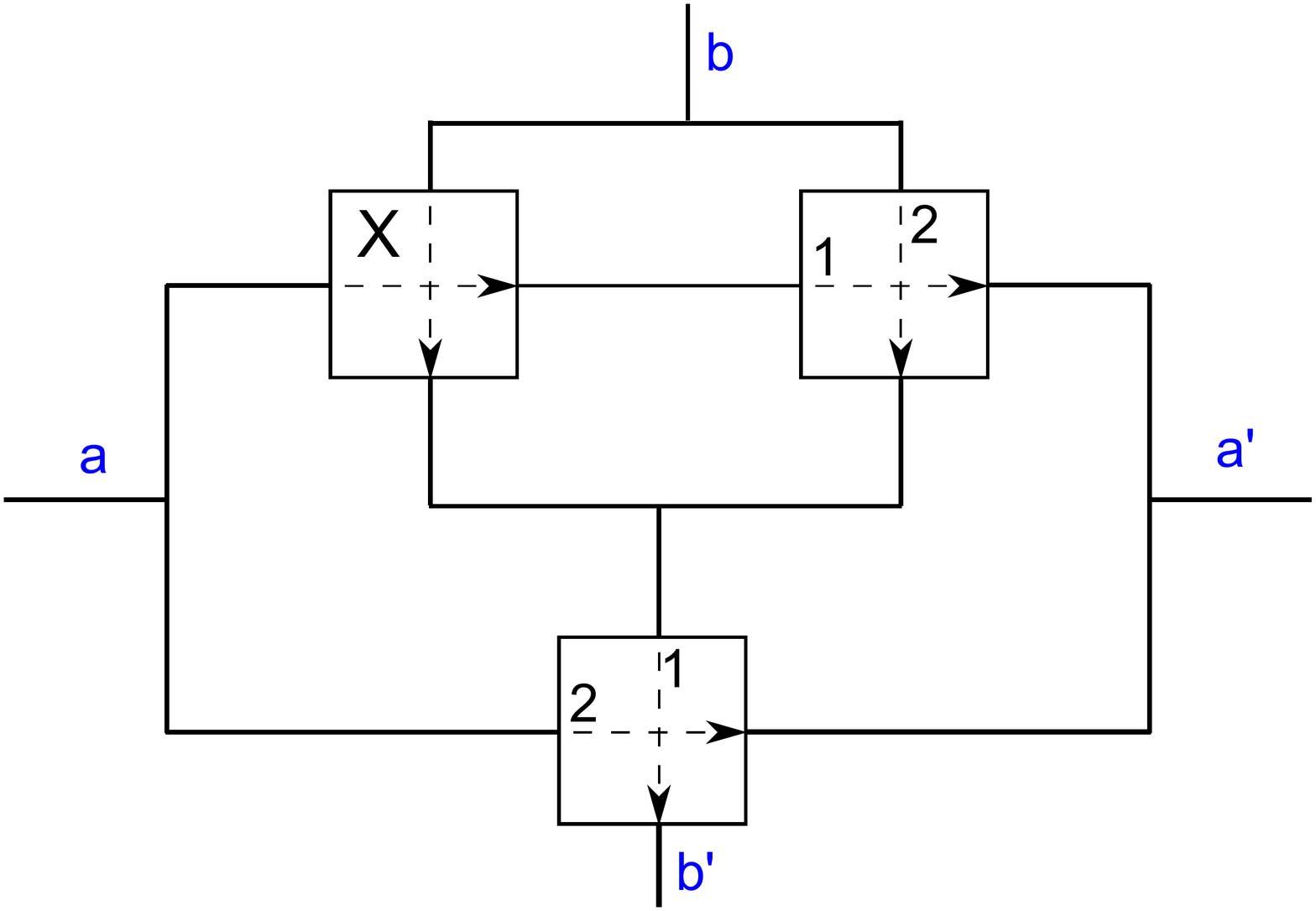}
    \caption{The two use directed crossover is constructed from a directed destructive crossover and two in-order directed crossovers.}
    \label{fig:OneUseCrossover}
    \vspace{-10mm}
\end{wrapfigure}

Now, there are at most 2 transitions possible: from $a'$ back to $a$, undoing the whole process, or entering at $b$. Note that entering at $b$ is only possible if the left Set-Verify is in the Set state, so let us assume that state change occurred. In that case, the left $S_o \rightarrow S_i$ transition may be performed, changing the left Set-Verify's state to Unset. At that point, the only possible transitions are back to $b$, or through the right Set-Verify's
$V_i \rightarrow V_o$ transition, completing the second crossover.

If the left Set-Verify was left in the Verify state, strictly less future transitions are possible compared to the case where it was changed into the set state, so we may disregard that possibility.

\paragraph{Two Use Directed Crossover} 
The Two Use Directed Crossover, depicted in Figure~\ref{fig:OneUseCrossover}, is the gadget needed for our proof. It allows a traversal from $a$ to $a'$ followed by a traversal from $b$ to $b'$, or from $b$ to $b'$ and then $a$ to $a'$. These transitions may also be reversed.

It is constructed out of an In-order Directed Crossover gadget and a Destructive Directed Crossover, as shown in Figure~\ref{fig:OneUseCrossover}. The $a$ to $a'$ traversal is initially passable, and goes through both gadgets,
blocking the destructive crossover but leaving the in-order crossover open for the $b$ to $b'$ traversal. If the $a$ to $a'$ traversal does not occur, the $b$ to $b'$ traversal is possible via the destructive crossover.

\begin{wrapfigure}{hrt}{0.45\textwidth}
\vspace{-5mm}
  \centering
    \includegraphics[width=.4\textwidth]{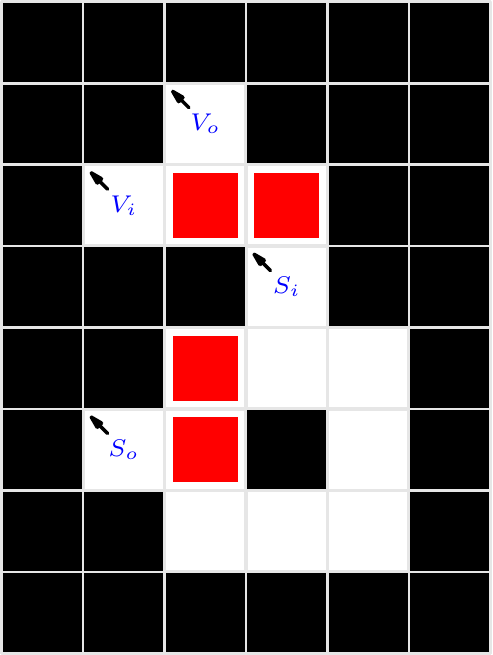}
    \caption{A Set-Verify gadget in 3D where the entrances and exits extend upward, notated by the diagonal arrows. This gadget is in the unset state.}
    \label{fig:3DSetVerify}
    \vspace{-10mm}
\end{wrapfigure}

In our construction with Set-Verify gadgets, the crossover also permits returning backwards through the crossings. However, these crossings can only be reversed in the same order they were made and never allow leaking into the other line. Thus no additional progress can be made via reversals, and so the gadget still works for NP-hardness.
\paragraph{Two Use Crossover} 
Four Directed Crossovers can be combined, as shown below, to create a crossover that can be traversed in any direction \cite{Push100}. This is not necessary for our proof but is shown for general interest. Unfortunately, the inability to go through this gadget multiple times in the same direction without first going back through means it likely isn't sufficient for PSPACE-completeness.

\begin{theorem}
\label{thm:2DNPhard}
Push-$k$ Pull-$l$ in 2D with thin walls is NP-hard.
\end{theorem}
\begin{proof}
    We will reduce from 3SAT. Given a 3SAT instance with variables $(x_1, x_2, \ldots x_n)$ and clauses $(x_a, x_b, \overline x_c), \ldots$, we will construct an equivalent PushPull instance as follows: 

    First, we will set up the clause gadgets. Each clause gadget will look like Figure~\ref{fig:NPClauseGadget}, with all of the Set-Verify gadgets initially in the unset state. There will be one clause gadget for each clause in the 3SAT formula. The clauses will be linked together in series, $C_k$ $out$ to $C_{k+1}$ in. At the final clause gadget's exit, we will place the goal square.

    Next, we will set up the variable gadgets. For each variable $x_k$, there will be a variable gadget $X_k$, consisting of a positive literal pathway, connecting to every clause where the variable is used positively, and a negative literal pathway, connecting to every clause where the variable is negated, as shown in Figure~\ref{fig:NPVariableGadget}. These variable gadgets will be linked together in series, $X_k$ $out$ to $X_{k+1}$ $in$. The final variable gadget's $out$ exit will be linked to the first clause gadget's $in$. Just in front of the first variable gadget's $in$ entrance will be the start square.

    The connections between these gadgets will consist of empty hallways, except where such hallways would cross. The hallways inside the clause and variable gadgets will also need to cross, and we will handle them similarly. We need crossovers for this reduction, rather than reducing to a PlanarSAT variant, because we need crossovers just to make the clause gadgets work.
    
    At all crossings, we will place a Two Use Directed Crossover, from Figure~\ref{fig:OneUseCrossover}. The orientation of the gadget will be chosen according to a specified ordering, where the later pathway will never be used before the earlier pathway, and no pathway will every be traversed twice in the same direction. The ordering is each variable gadget's hallways, in increasing order of the variable gadgets, followed by each clause gadget's hallways, in increasing order of clause gadgets. Within the variable gadgets, the ordering will be from in to out along the positive and negative lines, with the positive lines arbitrarily placed before the negative lines. The clause gadget hallways won't cross each other.

    The construction is complete. To see that it is solvable if and only if the corresponding SAT problem is satisfiable, first let us consider the case where the SAT problem is satisfiable. If the SAT problem is satisfiable, then there is an assignment of variables such that each clause is satisfied, e.g. has at least one true literal. Therefore, the PushPull construction is solvable. It can be solved by traversing each variable gadget via the side corresponding to the satisfying assignment, then traversing each clause, which is passable because it is satisfied. The crossovers do not impede traversal, since the path taken goes through each crossover at most once of each of its pathways, and strictly in the forward direction of the ordering which determined the orientation of the crossovers. Thus, the entire PushPull problem can be solved, as desired.

    Next, let us consider the case where the SAT problem is not satisfiable. Consider a partial traversal of the PushPull problem, from the start cell through the variable gadgets. Regardless of any reverse transitions through a variable gadget or interactions with its clause gadget, if the robot is beyond a given variable gadget exactly one of the variable lines must be set and the other must be unset. Likewise, the interactions with the crossover gadgets do not allow any transitions other than within the variable gadgets, regardless of reversals. Moreover, interactions with the clause gadgets only change the state of Set-Verify gadgets corresponding to literals between the Set and Verified states. If a Set-Verify is Unset, its state cannot be altered via its verify line ($V_i - V_o$).

    Thus, regardless of the robot's prior movements, the only literals that will be Set or Verified are at most those corresponding to a single assignment for each variable. No two literals corresponding to opposite assignments of the same variable will every be in the Set or Verified states at the same time. 

    Since the SAT problem is assumed to be unsatisfiable, no assignment of variables will satisfy every clause. Thus, as the robot exits the variable gadgets and enters the clause gadgets, for any prior sequence of moves, there must be some clause gadget which has all of its literals in the Unset state, corresponding to the unsatisfied clause for this setting of variables. Since all clauses must be traversed to reach the goal cell, and a clause cannot be traversed if all of its literals are Unset, the robot cannot reach the goal cell. Thus, the PushPull problem is unsolvable.

    We have demonstrated that the PushPull problem is solvable if and only if the corresponding 3SAT instance is satisfiable. The reduction mentioned above is polynomial time reduction, as long as the hallways are constructed reasonably. Thus, Push-$k$ Pull-$l$ in 2D with thin walls is NP-hard.

\end{proof}

\subsection{3D Push-Pull is NP-hard}
\label{3DNPhard}
In this section we prove that 3D Push-$k$ Pull-$l$ with fixed blocks is NP-hard, for all positive $k$ and $l$. All of the hard work was done in the previous section. Here we will simply show how we can use the additional dimension to tweak the previous gadgets to build them without thin walls. We reduce from 3SAT, constructing our variables from chains of 3D Set-Verify gadgets, and our clauses from the verify side of the corresponding 3D Set-Verify gadget.

\begin{theorem}
3D Push-$k$ Pull-$l$ with fixed blocks is NP-hard, for all positive $k$ and $l$.
\end{theorem}
\begin{proof}
We follow the proof of Theorem~\ref{thm:2DNPhard} using a modified Set-Verify gadget, shown in Figure~\ref{fig:3DSetVerify}.  It can be easily checked that this has the same properties as the Set-Verify given in Section~\ref{sec:SetVerifyGadgets}. We do note that the cyclic ordering of the entrances in the 3D Set-Verify is different from that of the 2D Set-Verify, however this is not important as we no longer need to construct crossovers. With a functional Set-Verify gadget, the remaining constructions of variables and clauses proceeded as in Section \ref{sec:2DPushPull3SAT}. No crossover gadgets are needed since we are working in 3D. Finally, we note that all blocks are in hallways of length at most 3, thus the gadgets still function as described for any positive push and pull values.

As before, in the unset state the only possible traversal is $S_i$ to $S_0$. This traversal allows the top right bock to be pulled down, moving the gadget into the set state. From here the $V$ to $V_0$ traversal is possible, as well as going back through the $S_0$ to $S$ pathway. However, the $S$ to $S_0$ traversal is not possible.

Variables are composed of hallways of 3D Set-Verify gadgets connected $S_0$ to $S$, one for each clause in which the variable appears, as in Figure~\ref{fig:NPVariableGadget}. Clauses are composed of three 3D Set-Verify gadgets connected in parallel as in Figure~\ref{fig:NPClauseGadget}. The details of these constructions follow those in Section~\ref{sec:2DPushPull3SAT} This completes the reduction from 3SAT. In addition, we note that all blocks are in hallways of length at most 3, thus the gadgets still function as described for any positive push and pull values.
\end{proof}

%
\section{PSPACE}
\label{3DPSPACE}
In this section we show the PSPACE-completeness of 3D push-pull puzzles with equal push and pull strength. We will prove hardness by a reduction from True Quantified Boolean Formula (also known as TQBF and 3QSAT), which asks whether, given a set of variables $\{x_1, x_2, \ldots x_n, y_1, y_2, \ldots y_n\}$ and a boolean formula $\theta(x_1, \ldots x_n, y_1 \ldots y_n)$ in conjunctive normal form with exactly three variables per clause, the quantified boolean formula $\forall y_1 \exists x_1 \forall y_2 \exists x_2 \ldots \theta(x_1, \ldots x_n, y_1, \ldots y_n)$ is true.\footnote{It is plausable a simpler reduction from Nondeterministic Constraint Logic\cite{GPCBook09} exists, since that problem has the advantage of also having an undirected state-space graph. However, its global (rather than agent based) was problematic and we failed to find such a reduction.}

We introduce a gadget called the 4-toggle and use it to simulate 3QSAT\cite{NPBook}. We construct the 4-toggle gadget in 3D push-pull block puzzles, completing the reduction. In particular we prove 3D Push1-Pull1 with thin walls is PSPACE-complete and 3D Push$i$-Pull$j$, for all positive $i=j$, is PSPACE-complete. A gap between NP and PSPACE still remains for 3D puzzles with different pull and push values, as well as for 2D puzzles. 


\subsection{Toggles}
We define an $n$-toggle to be a gadget which has $n$ internal pathways and can be in one of two internal states, $A$ or $B$. Each pathway has a side labeled $A$ and another labeled $B$. When the toggle is in the $A$ state, the pathways can only be traversed from $A$ to $B$ and similarly in the $B$ state it can only be traversed from $B$ to $A$. Whenever a pathway is traversed, the state of the toggle flips.

\begin{figure}[!ht]
\centering
\begin{subfigure}[t]{0.45\textwidth}
  \centering
    \includegraphics[width=0.8\textwidth]{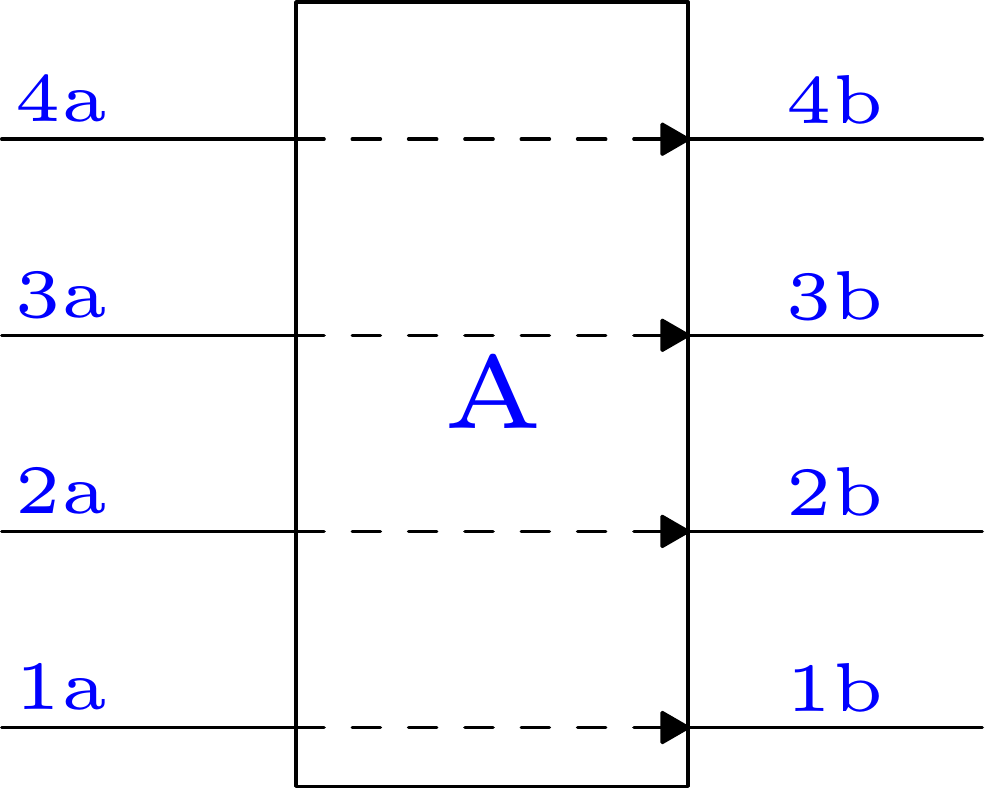}
    \caption{4-Toggle in state $A$. }
    \label{fig:Abstract4ToggleA}
\end{subfigure}
\begin{subfigure}[t]{0.45\textwidth}
  \centering
    \includegraphics[width=0.8\textwidth]{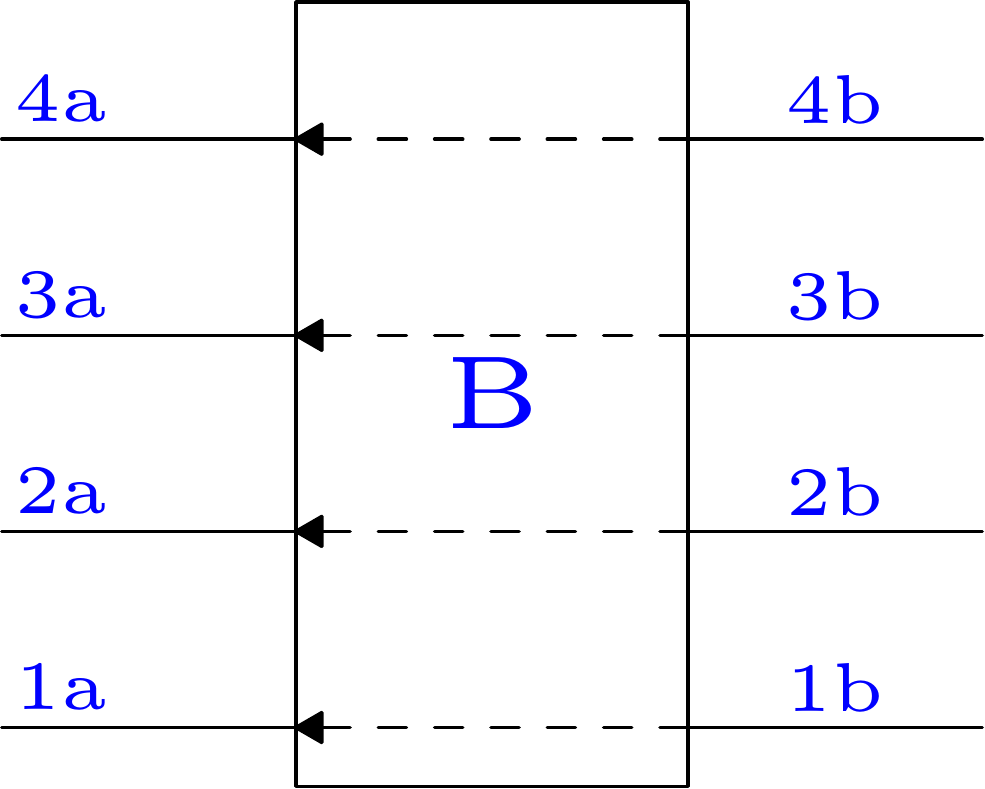}
    \caption{4-Toggle in state $B$. } 
    \label{fig:Abstract4ToggleB}
\end{subfigure}
\caption{Diagrams of the two possible states of a 4-Toggle.}
\label{fig:4-toggle}
\end{figure}

\begin{figure}[!ht]
\centering
\begin{subfigure}[t]{0.45\textwidth}
  \centering
    \includegraphics[width=0.8\textwidth]{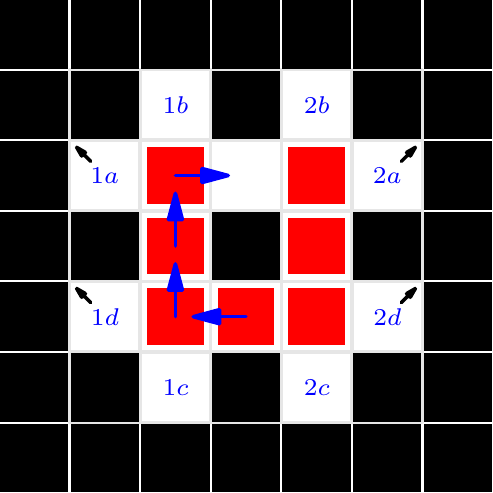}
    \caption{2-Toggle in state $A$. The arrows indicate the transition to state $B$.}
    \label{fig:2toggleA}
\end{subfigure}
\begin{subfigure}[t]{0.45\textwidth}
  \centering
    \includegraphics[width=0.8\textwidth]{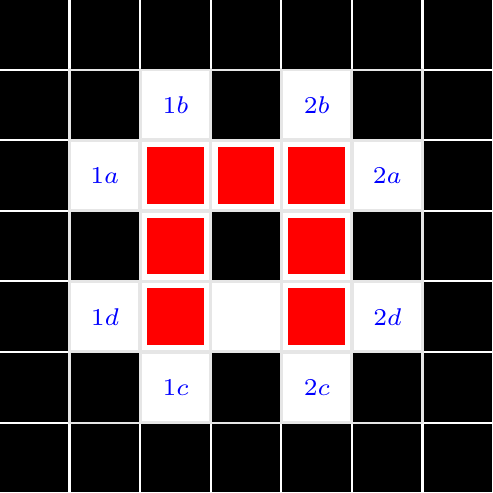}
    \caption{2-Toggle in state $B$.}
    \label{fig:2toggleB}
\end{subfigure}
    \begin{subfigure}[t]{0.45\textwidth}
  \centering
    \includegraphics[width=0.8\textwidth]{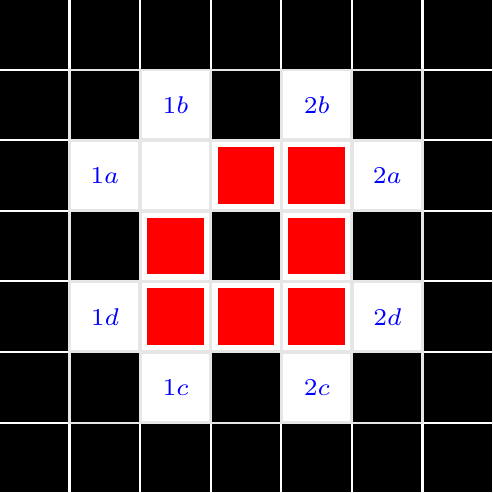}
    \caption{2-Toggle in one of four broken states.}
    \label{fig:broken2toggle}
    \end{subfigure}
  \begin{subfigure}[t]{.45\textwidth}
  \centering
    \includegraphics[width=.8\textwidth]{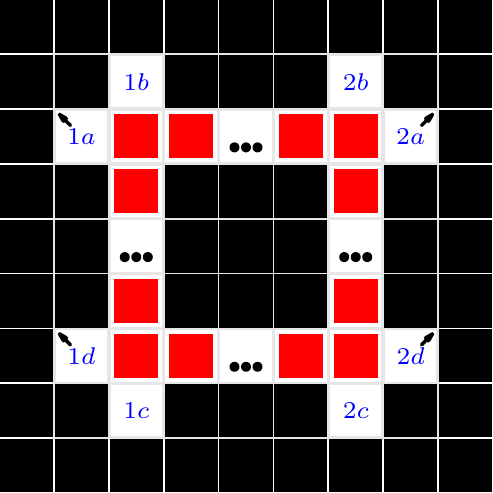}
    \caption{Construction of a 2-toggle when the robot can push or pull multiple blocks.}
    \label{fig:2ToggleK}
    \end{subfigure}
    \caption{2-Toggles constructed in a push-pull block puzzle.}
\end{figure}

\begin{table}
\begin{minipage}{.45\textwidth}
\centering
{\setlength\tabcolsep{4pt}
\begin{tabular}{>{$} l <{$} >{$} l <{$} >{$} l <{$} >{$} l <{$}}
   A: & & & \\
   &(A, 1a)& \rightarrow& (B, 1b) \\
   &(A, 2a)& \rightarrow & (B, 2b) \\
   &(A, 3a)& \rightarrow& (B, 3b) \\
   &(A, 4a)& \rightarrow & (B, 4b)  \\ \\
\end{tabular}}
\end{minipage}
\begin{minipage}{.45\textwidth}
\centering
{\setlength\tabcolsep{4pt}
\begin{tabular}{>{$} l <{$} >{$} l <{$} >{$} l <{$} l}
   B: & & & \\
   &(B, 1b)& \rightarrow& (A, 1a) \\
   &(B, 2b)& \rightarrow & (A, 2a) \\
   &(B, 3b)& \rightarrow& (A, 3a) \\
   &(B, 4b)& \rightarrow & (A, 4a)  \\ \\
\end{tabular}}
\end{minipage}
\caption{State transitions of a 4-toggle as seen in Figure~\ref{fig:4-toggle}}
\label{4ToggleStateTransition}
\end{table}

\begin{figure}[!ht]
  \centering
  \begin{subfigure}[t]{.45\textwidth}
    \includegraphics[width=\linewidth]{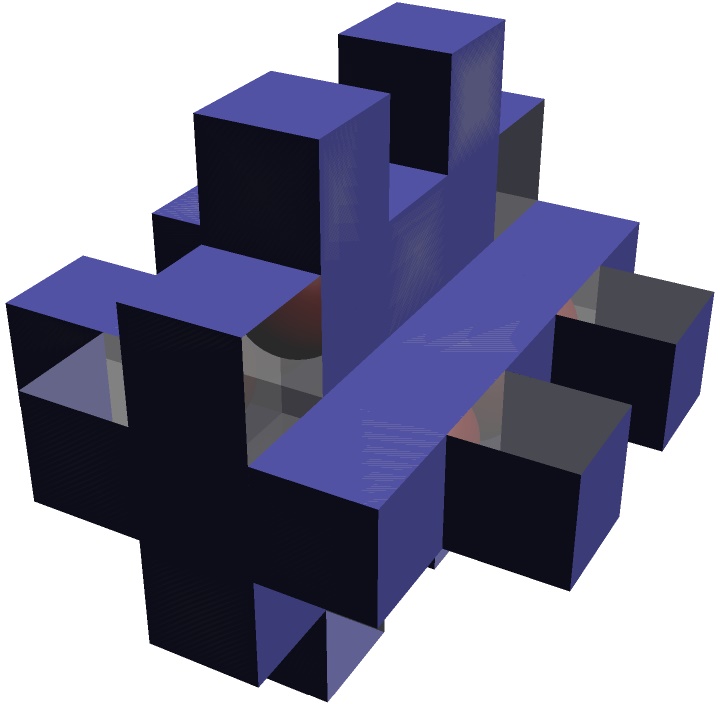}
    \caption{Diagram of a 4-toggle showing impassible surfaces.}
    \label{fig:4Toggle3D}
  \end{subfigure}
  \hfill
  \begin{subfigure}[t]{.45\textwidth}
    \includegraphics[width=\linewidth]{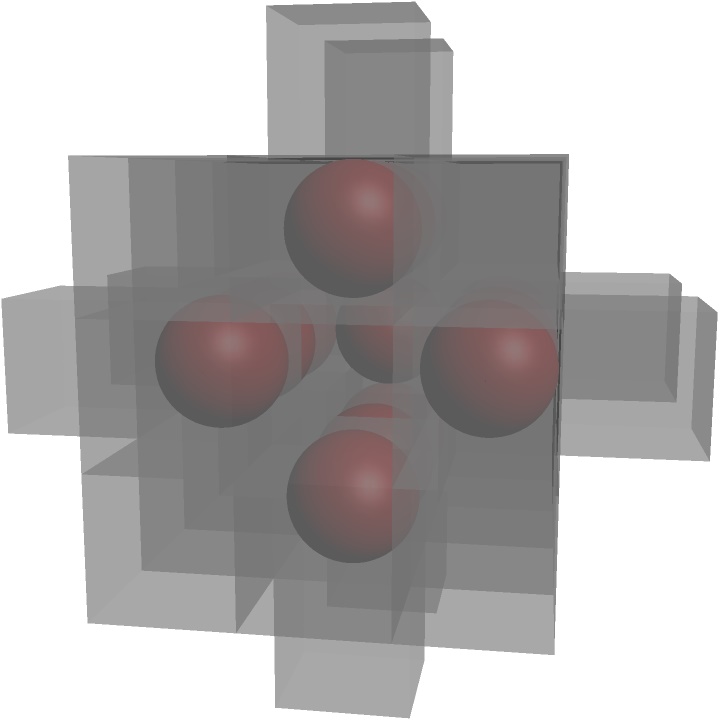}
    \caption{Diagram of the internals of a 4-toggle.}
    \label{fig:4Toggle3DBasic}
  \end{subfigure}
  \caption{3D diagrams of 4-toggles. Red spheres are blocks and blue surfaces are impassable.}
\end{figure}

Figure~\ref{fig:2toggleA} acts as a 2-toggle. The locations $1a$, $1d$, $2a$, and $2d$, are all entrances and exits to the 2-toggle, while $1b$ connects directly to $1c$, and $2b$ connects directly to $2c$. Notice that there is a single block missing from the ring of eight blocks. When the missing block is on top, as diagrammed, it will represent state $A$, and when it is on the opposite side, we call it state $B$. Notice that in state $A$, it is impossible to enter through entries $1d$ or $2d$. When we enter in the $1a$ or $2a$ sides, we can follow the moves in the series of diagrams to exit the corresponding $1d$ or $2d$ side, leaving the gadget in the $B$ state. One can easily check that the gadget can only be left in either state $A$, $B$, or a broken state with the empty square left in a corner. Notice that in the broken state, every pathway except the one just exited is blocked. If we enter through that path, it is in exactly the same state as if it had been in an allowed state and entered through the corresponding pathway normally. For example, in the diagram one can only enter through $1a$ and after doing so the blocks are in the same position as they would be after entering in path $1a$ on a 2-toggle in state $A$. Thus the broken state is never more useful for solving the puzzle and can be safely ignored. To generalize to Push-$k$ Pull-$k$ we simply expand the number of blocks between entrances and exists. Instead of having $3$ blocks between each entrance and exit, we have $2k+1$ blocks. There is still one vacant square left in the center of one of the rows of blocks to dictate the state of the toggle. The robot can push the row of $k$ blocks to the center or pull $k$ blocks opening up a square in the center, giving us the same function as before.

To construct a 4-toggle we essentially take two copies of the 2-toggle, rotate them perpendicular to each other in 3D, and let them overlap on the central axis, where the block is missing. See Figure~\ref{fig:4Toggle3D}. We still interpret the lack of blocks in the same positions as in the 2-toggle as states $A$ or $B$. Now we have four different paths which function like the ones described above. Similar arguments show the broken states of the 4-toggle don't matter.
For Push1-Pull1, this construction requires thin walls, since the exit pathways from $1b$, $2b$, $3b$ and $4b$ must pass immediately next to each other. For Push$k$-Pull$k$, with $k > 1$, thin walls are not necessary, since the exit pathways are separated from each other.
\subsection{Locks}

A 2-toggle and lock is a gadget consisting of a 2-toggle and a separate pathway. Traversing the separate pathway is only possible if the 2-toggle is in a specific state, and the traversal does not change the internal state of the 2-toggle. The 2-toggle functions exactly as described above.

This gadget can be implemented using a 4-toggle by
connecting the $3b$ and $4b$ entrances of the 4-toggle with an additional corridor, as shown in Figure~\ref{fig:LockA}.
Traversing the resultant full pathway, from $3a$ to $3b$ to $4b$ to $4a$, is possible only if the initial
state of the 4-toggle is $A$, and will leave the 4-toggle in state $A$. In addition, a partial traversal,
such as from $3a$ to $3b$ and back to $3a$, does not change the internal state. The two unaffected
pathways of the toggle, $1$ and $2$, continue to function as a 2-toggle.

\begin{wrapfigure}{tr}{0.45\textwidth}
\vspace{-5mm}
  \centering
    \includegraphics[width=.4\textwidth]{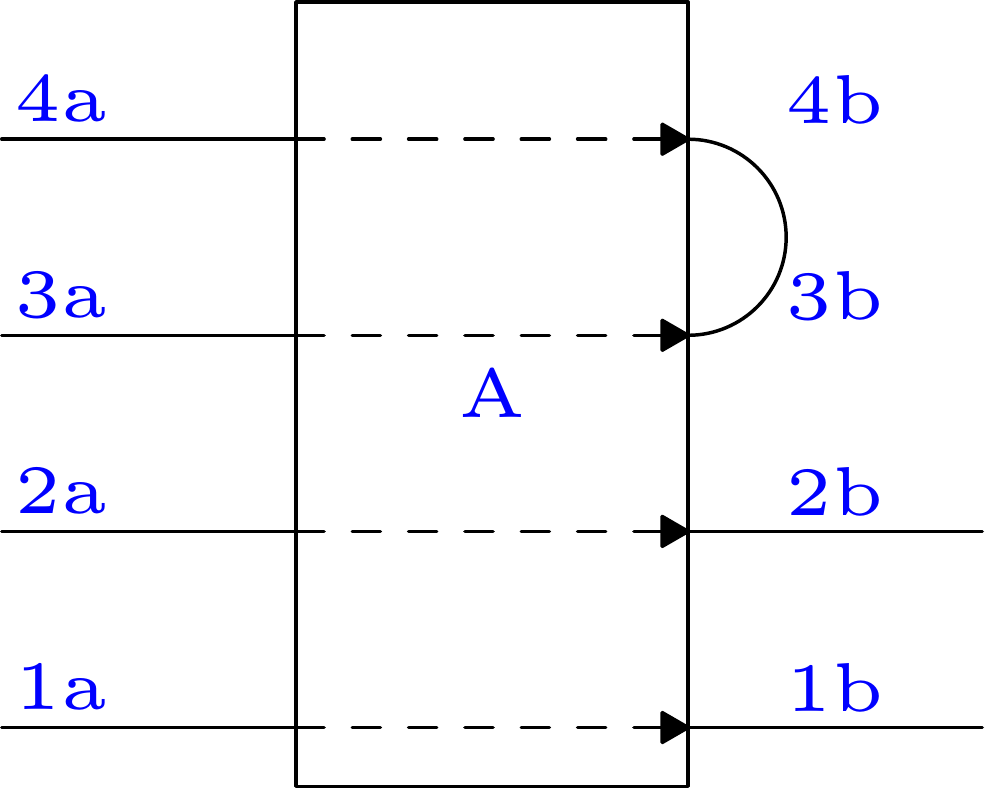}
    \caption{Diagram of a lock. The $3a$ to $4a$ traversal is only possible in state $A$ and returns the toggle to state $A$.}
    \label{fig:LockA}
\vspace{-7mm}
\end{wrapfigure}

A 2-toggle and lock can be extended to a 2-toggle with many locks. The 2-toggle with many locks is a gadget consisting of a 2-toggle and any number of separate pathways which can only be traversed when the gadget is in state $B$. This can be constructed using one 2-toggle and lock per separate pathway needed and attaching the toggles in series.
We orient the 2-toggles so that their 2-toggles are all passable at once in one direction.
When the 2-toggle is traversed, all of the internal locks' states flip, rendering the gadget passable in the opposite direction, and switching the passability and impassibility of all of the external pathways.

The 2-toggle still functions as described above.
Each pathway may be set up to be passable if and only if the toggle is in state $A$, or passable if and only if the toggle is in state $B$. 
This is implemented using one 2-toggle and lock per separate pathway needed and attaching the toggles in series.
By choosing 2-toggles and lock gadgets that are initially passable or initially unpassable, and orienting them all so that 
their 2-toggles are all passable at once, we can create a 2-toggle and many locks with an arbitrary arrangement of locks which correspond to the state of the 2-toggle.

This is implemented using one 2-toggle and lock per separate pathway needed. As shown in Figure~\ref{fig:LockBlock}, there are 2 pathways which run through the entire 2-toggle and many locks gadget: the 1 pathway, and the 2 pathway. 
By choosing 2-toggles and lock gadgets that are initially passable or initially unpassable, and orienting them all so that 
their 2-toggles are all passable at once, we can create a 2-toggle and many locks with an arbitrary arrangement of locks with 
either correspondence to the state of the 2-toggle.
\begin{figure}[b]
\centering
    \includegraphics[width=0.9\textwidth]{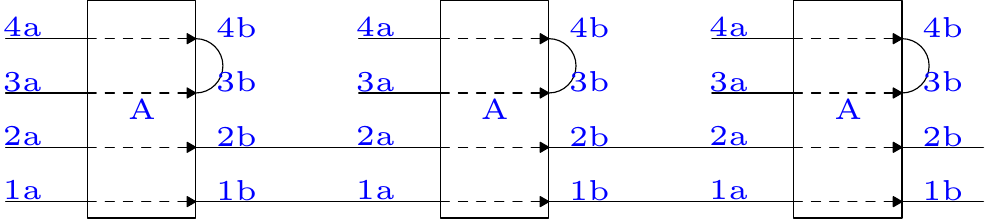}
    \caption{A 2-toggle and many locks.}
    \label{fig:LockBlock}
\end{figure}

When the 2-toggle is traversed, all of the internal locks' states flip, rendering the gadget passable in the opposite direction, and switching the passability and impassibility of all of the external pathways.

\subsection{Quantifiers}

We will construct a series of gadgets that will allow us to simulate existential and universal quantifiers in a boolean formula.

\subsubsection{Existential Quantifier}
An existential gadget is like a 2-toggle and many locks, except that instead of a
2-toggle, it has a single pathway which is always passable in both directions. Upon traversing the pathway
the robot may or may not change the internal state of the 2-toggle and many locks, as it chooses. The variable is 
considered true if the 2-toggles and many locks is in state $A$ and false if it is in state $B$. This gadget
is shown in Figure~\ref{fig:Existential}.

\begin{figure}[h!]
\centering
    \includegraphics[width=0.8\textwidth]{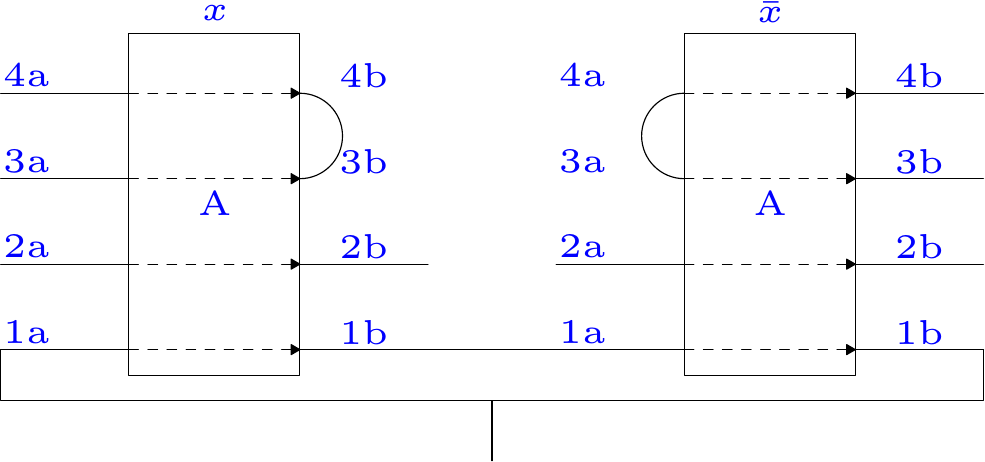}
    \caption{An existential gadget.} 
    \label{fig:Existential}
\end{figure}

\subsubsection{Binary Counter}
\label{sec:BinaryCounter}
Universal quantifiers must iterate through all possible combinations of values that they can take. In this section we construct a gadget that runs though all the states of its subcomponents as the robot progresses through the gadget. This construction will serve as the base for our universal quantifiers.

\begin{figure}[h!]
\centering
    \includegraphics[width=.9\textwidth]{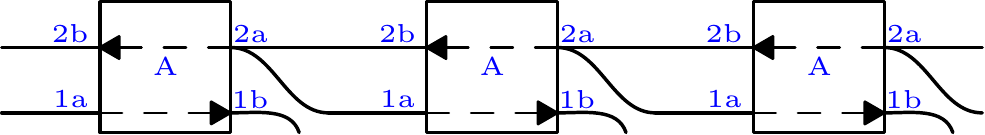}
    \caption{The central portion of a three bit binary counter made from 2-toggles.}
    \label{fig:BinaryCounter}
\end{figure}
  
A binary counter has a fixed number of internal bits.
Whenever the binary counter is traversed in the forwards direction, the binary number
formed by the internal bits increases by one and the robot leaves via one of the exits.
If the binary counter is traversed in the reverse direction, the internal value is reduced by
one. If the binary counter is partially traversed, but then the robot leaves via its initial entrance,
the internal value does not change.


The binary counter is implemented as a series of 2-toggles, as shown in Figure~\ref{fig:BinaryCounter}.
To see that this produces the desired effect, identify a toggle in state $A$ as a $0$ bit, and a toggle in state
$B$ as a $1$ bit. Let the entrance toggle's bit be the least significant bit, and the final toggle be the
most significant. When the robot enters the binary counter in the forwards direction, it will flip
the state of every toggle it passes through. When it enters a toggle that is initially in state $B$, and thus whose
bit is $1$, it will flip the state/bit and proceed to the next toggle, via the $2B - 2A$ pathway. When it
encounters a toggle that is initially in state $A$ / bit $0$, it will flip the state/bit and exit via the $1A - 1B$
pathway. Thus, the overall effect on the bits of the binary counter is to change a sequence of bits ending at the
least significant bit from $01\ldots11$ to $10\ldots00$, where the entrance is at the right.
This has the effect of increasing the value of the binary counter by one.
We will not examine the reverse transitions or rigorously complete the binary counter here, 
as we do not use it directly in the final construction. 

\subsubsection{Alternating Quantifier Chain}
\begin{figure}[h!]
\centering
    \includegraphics[width=0.9\textwidth]{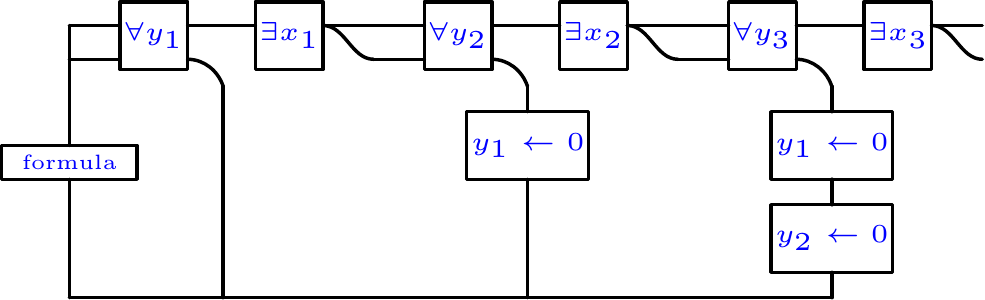}
    \caption{A segment of the alternating quantifier chain. Each square represents the 2-toggle part of a 2-toggle and many locks.}
    \label{fig:QuantifierChain}
\end{figure}

An alternating quantifier chain, shown in Figure~\ref{fig:QuantifierChain}, implements a series of
alternating existential and universal variables, as well as external literal pathways, which may be traversed
if and only if their corresponding variables are set to a prespecified value.

Traversing the quantifier chain
repeatedly in the primary direction will cycle the universal variables through all $2^n$ possible settings.
Upon each traversal, an initial sequence of the universal variables will have their values flipped.
During the traversal, the robot will have the option to set a series of corresponding existential variables to whatever value it wishes. These comprise the existentials nested within the universal variables whose values were flipped. An analysis of the Quantifier Chain can be found in Appendix~\ref{sec:AnalysisQuantifierChain}

Traversing  the quantifier chain in the reverse direction is only possible if the robot enters
via the lowest order universal toggle whose setting is $1$. The traversal will go back one setting in the
sequence of possible settings of the universal variables, and allow the robot to set all existential variables
corresponding to altered universal variables arbitrarily. No other existential variables can be changed.

There is also a special exit, the overflow exit, which can only be
reached after all of the universal variable settings have been traversed. This is the goal location for the robot.

A quantifier chain is implemented much like a binary counter (Section~\ref{sec:BinaryCounter}), with some additions. Every universal variable will be represented by a 2-toggle and many locks, where individual locks will serve as a literal. The 2-toggles are hooked up in the same manner as the 2-toggles in a binary counter gadget. This forces the 2-toggle and many locks gadgets to be set to the corresponding values in the simulated binary counter.

The next addition is the existential variables, which consist of existential gadgets placed just after the $2A$
exits of each universal variable,
and just before the $1A$ and $2B$ entrances of the next universal variable, as shown
in Figure~\ref{fig:QuantifierChain}.

One portion of the apparatus which has not been analyzed thus far is the potential for the robot to re-enter the chain of existentials
via a different exit pathway than the one just exited. This would be problematic if the robot re-entered via a universal gadget it had not just exited,
both because the robot should not be able to take any action other than reversing its prior progress, 
decrementing the binary counter/universal quantifiers. Problems would also arise if the robot got access to any existential quantifiers 
it did not just traverse.

After a traversal, the universal quantifiers have the settings $\ldots??10\ldots00$, where the lowest significance $1$ is 
on the pathway just exited.
To prevent the robot from re-entering via any pathway other than the one just exited, we add a series of locks to each exit that are only passable
if all lower-significance universal toggles are in state $0$, as shown in Figure~\ref{fig:QuantifierChain}.
This does not impede the exit that the robot uses initially, since all
lower-significance universal toggles are indeed $0$. These locks do prevent re-entry into any higher-significance universal toggles, since the
lock corresponding to the lowest-significance $1$ will be closed. The robot cannot re-enter via any toggle that is in state $0$,
due to the arrangement of the toggle pathways. Thus, the unique re-enterable pathway is the lowest-significance toggle in state $1$, as desired.

\subsection{Clause Gadget}
We construct a clause gadget by putting lock pathways of three 2-toggle with many locks in parallel, as we did with Set-Verify gadgets in Figure~\ref{fig:NPClauseGadget}. Each of these paths can be traversed only if the corresponding variable has been set to true, or to false, depending on the orientation of that particular lock. Since they are in parallel, only one needs to be passable for the robot to be able to continue on to the next clause.

\subsection{Beginning and End Conditions}
The overall progression of the robot through the puzzle starts with the quantifier chain.
The robot increments the universal variables and sets the appropriate existential variables arbitrarily, 
then traverses a series of clause gadgets to verify that the TQBF formula represented by those clauses 
is true under that setting of the variables. Then, the robot cycles around to the quantifier chain, and repeats.

At the beginning of this procedure, the robot must be allowed to set all of the existential variables arbitrarily.
To ensure this, we will set up the quantifier gadget in the state $01 \ldots 11$, with all variables set to $1$
except the highest order one.  The highest order variable will be special, and will not be used in the $3CNF$
formula. The initial position of the robot will be at the entrance to the quantifier gadget. This will allow
the robot to flip every universal in the quantifier gadget, from $01 \ldots 11$ to $10 \ldots 00$, and accordingly
set every existential variable arbitrarily. To force the robot to go forward through the quantifier gadget
instead of going backwards through the clause chain, we will add a literal onto the end of the formula gadget
which is passable if and only if the highest order variable is set to $1$.
After this set up, the robot will progress through the loop consisting of the quantifier gadget and the
formula gadget, demonstrating the appropriate existential settings for each assignment of the universal
quantifiers.

At each point in this process, the robot has the option to proceed through this cycle backwards, as is
guaranteed by the reversibility of the game. However, at no point does proceeding in the reverse direction
give the robot the ability to access locations or set toggles to states that it could not have performed
when it initially encountered the toggles or locations. Thus, any progression through the states of the
alternating quantifier chain must demonstrate a TQBF solution to the formula given.

After progressing through every possible state of the universal quantifiers, the universals will be in the
state $11 \ldots 11$. At this point, the robot may progress through the quantifier gadget and exit via its special
pathway,
the carry pathway of the highest order bit. This special pathway will lead to the goal location of the puzzle.
Thus, only by traversing the quantifier - formula loop repeatedly, and demonstrating the solution to the TQBF
problem,
will the robot be able to reach the goal. The robot may reach the goal if and only if the corresponding quantified
boolean formula is true.

\begin{theorem}
    \label{thm:3dPSPACE-complete}
    Push-$k$ Pull-$k$, $k>1$ in 3D with fixed blocks is PSPACE-complete.
\end{theorem}
\begin{proof}
By the above construction, TQBF can be reduced to Push-$k$ Pull-$k$ in three dimensions with fixed walls, through the intermediate step of construction a 4-toggle. This implies that Push$k$-Pull$k$ is PSPACE-hard.
Since Push$k$-Pull$k$ has a polynomial-size state, the problem is in NPSPACE, and therefore in PSPACE by Savitch's Theorem\cite{SAVITCH1970177}. So it is PSPACE-complete.
\end{proof}

\begin{theorem}
    Push-1 Pull-1 in 3D with thin walls is PSPACE-complete.
\end{theorem}
\begin{proof}
    Push-1 Pull-1 in 3D with thin walls can construct a 4-toggle, and so by the same argument as in Theorem~\ref{thm:3dPSPACE-complete}, it is PSPACE-complete.
\end{proof}

\section{Conclusion}
In this paper, we proved hardness results about variations of block-pushing puzzles in which the robot can also pull blocks. Along the way, we analyzed the complexity of two new, simple gadgets, creating useful new toolsets with which to attack hardness of future puzzles. The results themselves are obviously of interest to game and puzzle enthusiasts, but we also hope the analysis leads to a better understanding of motion-planning problems more generally and that the techniques we developed allow us to better understand the complexity of related problems.

This work leads to many open questions to pursue in future research. For Push-Pull block puzzles, we leave several NP vs.\ PSPACE gaps, a feature shared with many block-pushing puzzles. One would hope to directly improve upon the results here to show tight hardness results for 2D and 3D push-pull block puzzles. One might also wonder if the gadgets used, or the introduction of thin walls, might lead to stronger results for other block-pushing puzzles. We also leave open the question of push-pull block puzzles without fixed blocks or walls. In this setting, even a single $3\times3$ area of clear space allows the robot to reach any point, making gadget creation challenging.

There are also interesting questions with respect to the abstract gadgets introduced in our proof. We are currently studying the complexity of smaller toggles and toggle-lock systems. It would also be interesting to know whether Set-Verify gadgets sufficient for PSPACE-hardness or if they can build full crossover gadgets. Also, there are also many variations within the framework of connected blocks with traversibility which changes with passage through the gadget. Are any other gadgets within this framework useful for capturing salient features of motion planning problems? Finally, there is the question of whether other computational complexity problems can make use of these gadgets to prove new results.


\bibliographystyle{alpha}
\bibliography{PushPullBib}{}

\begin{thebibliography}{ADGV14}

\bibitem[AAD16]{abdelkader2048}
Ahmed Abdelkader, Aditya Acharya, and Philip Dasler.
\newblock 2048 without new tiles is still hard.
\newblock In {\em LIPIcs-Leibniz International Proceedings in Informatics},
  volume~49. Schloss Dagstuhl-Leibniz-Zentrum fuer Informatik, 2016.

\bibitem[ADGV14]{NintendoFun2014}
Greg Aloupis, Erik~D. Demaine, Alan Guo, and Giovanni Viglietta.
\newblock Classic {N}intendo games are ({NP}-)hard.
\newblock In {\em Proceedings of the 7th International Conference on Fun with
  Algorithms (FUN 2014)}, Lipari Island, Italy, July 1--3 2014.

\bibitem[Cul98]{Sokoban98}
J.~C. Culberson.
\newblock Sokoban is {PSPACE}-complete.
\newblock In {\em Proceedings International Conference on Fun with Algorithms
  (FUN 1998)}, pages 65--76, Waterloo, Ontario, Canada, June 1998. Carleton
  Scientific.

\bibitem[DDO00]{Push100}
Erik~D. Demaine, Martin~L. Demaine, and Joseph O'Rourke.
\newblock {PushPush} and {Push-1} are {NP}-hard in {2D}.
\newblock In {\em Proceedings of the 12th Annual Canadian Conference on
  Computational Geometry (CCCG 2000)}, pages 211--219, Fredericton, New
  Brunswick, Canada, August 16--18 2000.

\bibitem[DH01]{non-crossing01}
Erik~D. Demaine and Michael Hoffmann.
\newblock Pushing blocks is {NP}-complete for noncrossing solution paths.
\newblock In {\em Proceedings of the 13th Canadian Conference on Computational
  Geometry (CCCG 2001)}, pages 65--68, Waterloo, Ontario, Canada, August 13--15
  2001.

\bibitem[DHH02]{Push2F02}
Erik~D. Demaine, Robert~A. Hearn, and Michael Hoffmann.
\newblock Push-2-{F} is {PSPACE}-complete.
\newblock In {\em Proceedings of the 14th Canadian Conference on Computational
  Geometry (CCCG 2002)}, pages 31--35, Lethbridge, Alberta, Canada, August
  12--14 2002.

\bibitem[DHH04]{PushPushk04}
Erik~D. Demaine, Michael Hoffmann, and Markus Holzer.
\newblock Push{P}ush-$k$ is {PSPACE}-complete.
\newblock In {\em Proceedings of the 3rd International Conference on Fun with
  Algorithms (FUN 2004)}, pages 159--170, Isola d'Elba, Italy, May 26--28 2004.

\bibitem[DO92]{DO92}
A.~Dhagat and J.~O'Rourke.
\newblock Motion planning amidst movable square blocks.
\newblock In {\em Proceedings of the 4th Canadian Conference on Computational
  Geometry (CCCG 1992)}, 1992.

\bibitem[DZ96]{DZ96}
D.~Dor and U.~Zwick.
\newblock Sokoban and other motion planning problems.
\newblock {\em Computational Geometry: Theory and Applications}, 13(4), 1996.

\bibitem[FB02]{RushHour02}
Gary~William Flake and Eric~B. Baum.
\newblock Rush {H}our is {PSPACE}-complete, or why you should generously tip
  parking lot attendants.
\newblock {\em Theoretical Computer Science}, 270(1-2):895 -- 911, 2002.

\bibitem[GJ79]{NPBook}
Michael~R. Garey and David~S. Johnson.
\newblock {\em Computers and Intractability: A Guide to the Theory of
  {NP}-Completeness}.
\newblock W. H. Freeman \& Co., New York, NY, USA, 1979.

\bibitem[GLN14]{guala2014bejeweled}
Luciano Guala, Stefano Leucci, and Emanuele Natale.
\newblock Bejeweled, candy crush and other match-three games are (np-) hard.
\newblock In {\em Computational Intelligence and Games (CIG), 2014 IEEE
  Conference on}, pages 1--8. IEEE, 2014.

\bibitem[HD05]{hearn2005pspace}
Robert~A. Hearn and Erik~D. Demaine.
\newblock {PSPACE}-completeness of sliding-block puzzles and other problems
  through the nondeterministic constraint logic model of computation.
\newblock {\em Theoretical Computer Science}, 343(1):72--96, 2005.

\bibitem[HD09]{GPCBook09}
Robert~A. Hearn and Erik~D. Demaine.
\newblock {\em Games, {Puzzles}, and {Computation}}.
\newblock A. K. Peters, Ltd., Natick, MA, USA, 2009.

\bibitem[Hof00]{Push*00}
M.~Hoffman.
\newblock Push-* is {NP}-hard.
\newblock In {\em Proceedings of the 12th Canadian Conference on Computational
  Geometry (CCCG 2000)}, Lethbridge, Alberta, Canada, 2000.

\bibitem[Rit10]{Pull10}
Marcus Ritt.
\newblock Motion planning with pull moves.
\newblock arXiv:1008.2952, 2010.

\bibitem[RW86]{15Puzzle}
Daniel Ratner and Manfred~K Warmuth.
\newblock Finding a shortest solution for the n$\times$ n extension of the
  15-{PUZZLE} is intractable.
\newblock In {\em Proceedings of AAAI 1986}, pages 168--172, 1986.

\bibitem[Sav70]{SAVITCH1970177}
Walter~J. Savitch.
\newblock Relationships between nondeterministic and deterministic tape
  complexities.
\newblock {\em Journal of Computer and System Sciences}, 4(2):177 -- 192, 1970.

\bibitem[Wil91]{PushPull91}
Gordon Wilfong.
\newblock Motion planning in the presence of movable obstacles.
\newblock {\em Annals of Mathematics and Artificial Intelligence},
  3(1):131--150, 1991.

\bibitem[ZR11]{zubaranagent}
Tadeu Zubaran and Marcus Ritt.
\newblock Agent motion planning with pull and push moves.
\newblock In {\em Anais do VIII Encontro Nacional de Inteligência Artificial
  (ENIA)}, Natal, July 2011.

\end{thebibliography}
%

\end{document}